\documentclass[conference,letterpaper]{IEEEtran}
\IEEEoverridecommandlockouts
\usepackage{cite}
\usepackage{amsthm}
\newtheorem{theorem}{Theorem}
\newtheorem{definition}{Definition}
\newtheorem{corollary}{Corollary}[theorem]
\newtheorem{lemma}{Lemma}
\newtheorem{proposition}{Proposition}
\newtheorem*{lemma*}{Lemma}

\theoremstyle{definition}
\newtheorem{remark}{Remark}

\newcommand{\sumd}[1] { \sum_{i\in [#1]}}
\newcommand{\sumdi}[1] { \sum_{i\in [#1]}}
\newcommand{\sumdk}[1] { \sum_{k\in [#1]}}

\newcommand{\nstar} { \bar{n} }
\newcommand{\summis} { \sumdi{\nstar} }

\newcommand{\sumt}{{\sum_{t\in \mathcal F}}} 
\newcommand{\suml}{{\sum_{l\in \mathcal K}}}

\newcommand{\lp}{\left(} 
\newcommand{\rp}{\right)} 
\newcommand{\lb}{\left[} 
\newcommand{\rb}{\right]} 

\newcommand{\MPALGONAME}{\textsc{AMIX-MS}} 
\newcommand{\ALGONAME}{\textsc{AMIX-ND}} 

\newcommand{\E}{\mathbb E} 
\newcommand{\ER}{\mathbb E^R} 
\newcommand{\ERJ}{\mathbb E^{R,J}} 
\newcommand{\EJ}{\mathbb E^{J}}

\newcommand{\allMIS}{\mathcal M} 
\newcommand{\numbMIS}{|\allMIS|}

\newcommand{\users}{\mathcal K} 
\newcommand{\state}{\mathcal S} 
\newcommand{\statez}{\mathcal{S}(t_0)} 
\newcommand{\statet}{\mathcal{S}(t)} 

\newcommand{\maxpacks}{a_{\max}} 
\newcommand{\maxdead}{d_{\max}} 
\newcommand{\ifgraph}{G_I} 

\newcommand{\gain}{\mathcal G} 
\newcommand{\again}{\mathcal{G}^\prime} 
\newcommand{\frm}{\mathcal F}

\newcommand{\strat}{\mathcal{P}_{NC}(\mathcal{F})}
\newcommand{\cpolicy}{\mathcal P_C}
\newcommand{\ndpol}{\mathcal{P}_{ND}}
\newcommand{\nonempty}{\mathcal{B}(t)}

\newcommand{\extracompmul}{\mathcal E_m}

\newcommand{\be}{\begin{eqnarray}}
\newcommand{\ve}[1]{\mathbf{#1}}
\newcommand{\ee}{\end{eqnarray}}
\newcommand{\ben}{\begin{eqnarray*}}
\newcommand{\een}{\end{eqnarray*}}
\newcommand{\bfl}{\begin{flalign*}}
\newcommand{\efl}{\end{flalign*}}
\newcommand{\dref}[1]{(\ref{#1})}

\newcommand{\ADVR}{\hat{\mu}}
\newcommand{\ADV}{\mu}
\newcommand{\ALG}{ALG}

\newcommand{\MIS}{MIS}
\newcommand{\traffic}{\tau}
\newcommand{\nondominated}{\mathcal{B}_{\text{ND}}}

\newcommand{\subhrm}[1]{ C_{#1}(t)}
\newcommand{\misw}[1]{  W_{M_{#1}}(t)}

\newcommand{\MREMOVE}[1] {}

\newcommand{\NEW}[1] {#1}

\newcommand{\trafficstate}{\mathbf{x}}
\newcommand{\trafficextra}{u}
\usepackage{algorithm, algorithmicx, algpseudocode}

\usepackage{amsmath,amssymb,amsfonts}
\usepackage{dsfont,bm}
\usepackage{mathtools}
\DeclareMathOperator{\interior}{int}

\usepackage{enumitem}
\usepackage{epsfig}

\usepackage{graphicx}
\usepackage{textcomp}
\usepackage{xcolor}
\usepackage{subcaption}

\def\BibTeX{{\rm B\kern-.05em{\sc i\kern-.025em b}\kern-.08em
    T\kern-.1667em\lower.7ex\hbox{E}\kern-.125emX}}

\begin{document}

\DeclarePairedDelimiter\ceil{\lceil}{\rceil}
\DeclarePairedDelimiter\floor{\lfloor}{\rfloor}

\title{On the Power of Randomization for Scheduling Real-Time Traffic in Wireless Networks}
%Adaptive REMIX: A Better Algorithm than LDF for Real-Time Scheduling \\

\author{
	\IEEEauthorblockN{Christos Tsanikidis, Javad Ghaderi}
	\IEEEauthorblockA{Electrical Engineering Department, Columbia University}  
	%	Email: \{c.tsanikidis, jghaderi\}@columbia.edu}
	\thanks{Emails: \{c.tsanikidis, jghaderi\}@columbia.edu. This research was supported by grants NSF 1717867, NSF 1652115, and ARO W911NF1910379.}
	}

%\author{Submission \#1570577458}
%\IEEEauthorblockN{1\textsuperscript{st} Given Name Surname}
% \and

\maketitle

\begin{abstract}
In this paper, we consider the problem of scheduling real-time traffic in wireless networks under a conflict-graph interference model
and single-hop traffic. %Each packet has a strict deadline and will be dropped if it is not transmitted before the deadline. 
The objective is to guarantee that at least a certain fraction of packets of each link are delivered within their deadlines, which is referred to as \textit{delivery ratio}. This problem has been studied before under restrictive frame-based traffic models, or greedy maximal scheduling schemes like LDF (Largest-Deficit First)  that can lead to poor delivery ratio for general traffic patterns. In this paper, we pursue a different approach through randomization over the choice of maximal links that can transmit at each time. We design randomized policies in collocated networks, multi-partite networks, and general networks, that can achieve delivery ratios much higher than what is achievable by LDF. Further, our results apply to traffic (arrival and deadline) processes that evolve as positive recurrent Markov chains. Hence, this work is an improvement with respect to both efficiency and traffic assumptions compared to the past work. We further present extensive simulation results over various traffic patterns and interference graphs to illustrate the gains of our randomized policies over LDF variants. 
\end{abstract}

\section{Introduction}

Much of the prior work on scheduling algorithms for wireless networks focus on maximizing
throughput. However, for many real-time applications, e.g., in Internet of Things (IoT), vehicular networks, and other cyber-physical systems, delays and deadline guarantees on packet delivery are more important than long-term throughput~\cite{lu2015real, song2008wirelesshart, gubbi2013internet}. Recently, there has been an interest in developing scheduling algorithms specifically targeted
towards handling deadline-constrained traffic~\cite{houborkum09,houkum09,kumar10,srikant10,kang2014performance,kang2015capacity}, when each packet has to be delivered within a \textit{strict deadline}, otherwise it is of no use. The key objective in these works is to guarantee that at least a fraction of the packets will be delivered to their destinations within their deadlines, which is refereed to as \textit{delivery ratio} (QoS). Providing such guarantees is very challenging as it crucially depends on the temporal pattern of packet arrivals and their deadlines, as opposed to long-term averages in traditional throughput maximization. One can construct adversarial traffic patterns that all have the same long-term average but their achievable delivery ratio is vastly different~\cite{kang2014performance,reddy2012effect}.  

Recently, there have been two approaches for providing QoS guarantees for real-time traffic in wireless networks. One is the frame-based approach~\cite{houborkum09,houkum09,kumar10,srikant10}, and the other is a greedy scheduling approach like the largest-deficit-first policy (LDF)~\cite{kang2014performance,kang2015capacity}.
In the frame-based approach, it is assumed that each frame is a number of consecutive time slots, and packets arriving in each frame have to be scheduled before the end of the frame. They crucially rely on the assumption that all packets of all users arrive at the beginning of frames~\cite{houborkum09,houkum09,kumar10}, or   
the complete knowledge of future packet arrivals and their deadlines in each frame is available at the beginning of the frame~\cite{srikant10}. This restricts the application of such policies to specific traffic patterns with periodic arrivals and synchronized users. The results for general traffic patterns without such frame assumptions are very limited, as in such settings, the real-time rate region is difficult to characterize and the optimal policy is unknown. A popular algorithm for providing QoS guarantees for real-time traffic is the largest-deficit-first (LDF) policy~\cite{kang2014performance,kang2015capacity,houborkum09, jaramillo2011scheduling}, which is the real-time variation of the longest-queue-first (LQF) policy (see, e.g.,\cite{joo2009understanding,dimakis2006sufficient}).
%\footnote{LDF selects a link with the largest deficit, removes its interfering links, and repeats the procedure.}
It is known that LDF is optimal in collocated networks \textit{under the frame-based model}~\cite{houborkum09, jaramillo2011scheduling}. The performance of LDF in the non-frame-based setting has been studied in~\cite{kang2014performance} in terms of the \textit{efficiency ratio}, which is the fraction of the real-time throughput region guaranteed by LDF. It is shown that LDF achieves an efficiency ratio of at least $\frac{1}{1+\beta}$ for a network with interference degree\footnote{The interference degree is the maximum number of links that can be scheduled simultaneously out of a link and its neighboring links.} $\beta$, under \textit{i.i.d.} (independent and identically distributed) packet arrivals and deadlines. Further, when traffic is not i.i.d., the efficiency ratio of LDF is as low as $\frac{1}{1+\sqrt{\beta}}$~\cite{kang2014performance}. In particular, for collocated networks, the efficiency ratio of LDF
under non-i.i.d. traffic is $1/2$, and in a simple star topology with one center link and $K$ neighboring links, it scales down as low as $O(\frac{1}{\sqrt{K}})$. This shows that LDF might not be suitable for high throughput real-time applications, especially with non-i.i.d. traffic, which is the case if packet drops due to deadline expiry trigger re-transmissions.   

% As a motivating example, consider a simple scenario with two interfering wireless links. Hence, at any time slot, only one of the links can transmit. Suppose at the beginning of time slots $t=0,2,4,\cdots$, one packet arrives to each link, however packets of link $1$ have a deadline of one slot and packets of link $2$ have a deadline of two. Clearly, an optimal scheduler should schedule packet of link $1$ first and then link $2$, so that all the packets will be scheduled. Now consider a greedy policy like LDF that does not consider the deadlines.
% %, and only reacts to deficits, e.g., a measure of how many packets each link has lost. In this case, when deficits are equal, 
% In this case, LDF can choose link $2$ to transmit, thus dropping the packet of link $1$. In Section~\ref{sec:col}, we present a detailed traffic pattern where LDF cannot schedule more than $1/2+\epsilon$ of packets from each link, for any $\epsilon>0$.

Besides the works above on providing QoS guarantees for wireless networks, there is literature on approximation algorithms for \textit{single-link} buffer management problem~\cite{hajek2001competitiveness,kesselman2004buffer}. In this problem, packets arrive to a single link, each with a non-negative constant weight and a deadline. The goal is to maximize the total weight of transmitted packets for the worst input sequence. The approximation algorithms include the maximum-weight greedy algorithm~\cite{hajek2001competitiveness,kesselman2004buffer}, EDF$_{\alpha}$~\cite{chin2006online} which schedules the earliest-deadline packet with weight at least $\alpha\geq1$ of the maximum-weight packet, or randomized algorithms such as~\cite{bienkowski2011randomized,jez2011one,jez2013universal} where the scheduling decision is randomized over pending packets in the link's buffer. Inspired by such randomization techniques, we design randomized algorithms for wireless networks under a general interference model and given the delivery ratio requirements for the links in the network. 
%We show how randomized scheduling schemes can extend to general wireless networks in order to guarantee the required delivery ratios.    

%\textit{A Star Network.}

% In increasingly more contemporary applications of wireless-networks, we are dealing with traffic that is delay sensitive. The packets, depending on the application, might have reduced utility as their age increases \TODO{age of information?} or they could become entirely irrelevant after a hard-deadline \TODO{CITE papers where this happens}. In this paper we are examining the latter case and we are trying to provide guarantees on the delivery ratio. This model could capture a network for example the case of real-time streaming, where if a frame is not transmitted immediately, it should be discarded.

%\textit{Previous Work and the problems.}
\subsection{Contributions} 
%The contributions of this work can be summarized as follows.
%In summary, our constributions can be 
\textbf{Non-i.i.d. (Markovian) Traffic Model.} Our traffic model allows traffic (\textit{arrival} and \textit{deadline}) processes that evolve as an irreducible Markov chain over a finite state space. This model is a significant extension from i.i.d. or frame-based traffic models in~\cite{houborkum09,houkum09,kumar10,srikant10,kang2014performance}. A key technique in analyzing the achievable efficiency ratio in our model is to look at the return times of the traffic Markov chain and analyze the performance of scheduling algorithms over long enough cycles consisting of multiple return times. 

\textbf{Randomized Algorithms with Improved Efficiency.} We propose randomized scheduling algorithms that can significantly outperform deterministic greedy algorithms like LDF. The key idea is to identify a structure for the optimal policy and randomize over the possible scheduling choices of the optimal policy, rather than  solely relying on the deficit queues. For \textit{collocated networks} and \textit{complete bipartite graphs} our randomized algorithms achieve an efficiency ratio of at least $0.63$ and $2/3$, respectively, and in \textit{general graphs}, achieve an efficiency ratio of at least $1/2$, all \textit{independent of the network size} and \textit{without the knowledge of the traffic model}.

\section{Model and Definitions}
\textit{Wireless Network Model.} We consider a set of $K$ links (or users) denoted by the set $\users$, where $K=|\users|$. 
Time is slotted, and at each time slot $t\in \mathbb N_0$, each link can transmit one packet successfully, if there are no interfering links transmitting at the same time. As in~\cite{kang2014performance}, it is standard to represent the interference relationships between links by an \textit{interference graph} $\ifgraph = (\users,E_I)$. Each vertex of $\ifgraph$ is a link, and an edge $(l_1,l_2)\in E_I$ indicates links $l_1$ and $l_2$ interfere with each other. Let $I_l(t)=1$ if link $l$ is transmitting a packet at time $t$, and $I_l(t)=0$ otherwise.
Hence, at any time any feasible schedule $I(t)=(I_l(t), l \in \users)$ has to form an independent set of $\ifgraph$ over links that have packets, i.e., no two transmitting links can share an edge in $\ifgraph$. We say a feasible schedule $I$ is maximal if no more links can be scheduled without interfering with some active links in $I$. Let $\nonempty$ be the set of links that have packets available to transmit at time $t$. Let $\allMIS$ denote the set of all maximal independent sets of $\ifgraph$. Then, at any time $t$, 
\ben
&\{l \in \users: I_l(t)=1\} \subseteq (\nonempty \cap M), \text{ for some }M \in \allMIS,
\een
where `$\subseteq$' holds with `$=$' if $I$ is a maximal schedule.  

\textit{Traffic Model.} We consider a single-hop traffic with deadlines for each link. Let $a_l(t)$ denote the number of packets arriving on link $l$ at time $t$,
%and $A_l(t) = \sum_{s=1}^{t} a_l(s)$ be the total number of arrivals up to (and including) time slot $t$.
with $a_l(t)\leq \maxpacks$, for some $ \maxpacks < \infty$. 
Each packet upon arrival has a deadline which is the maximum delay that the packet can tolerate. We define a vector $\tau_{l}(t)=(\tau_{l,d}(t); d=1,\cdots,\maxdead)$, where $\tau_{l,d}(t)$ is the number of packets with deadline $d$ arriving to link $l$ at time $t$. A packet arriving with deadline $d$ at time $t$ has to be transmitted before the end of time slot $t+d-1$, otherwise it will be dropped. The maximum deadline is bounded by a constant $\maxdead$. \NEW{Hence, the network traffic (arrival, deadline) process is described by $\tau(t)=(\tau_l(t); l \in \users), \ t \geq 0$. 
%Each packet upon arrival is associated with a deadline 
%Further we define the total packets having arrived by time $t$: $A_l(t) = \sum_{s=1}^{t} a_l(s)$. 
We also use $\trafficextra(t)$ to denote any unobservable (hidden) information of the traffic process, so that  the complete traffic process $\trafficstate(t)=(\tau(t), \trafficextra(t))$ evolves as an irreducible Markov chain over a finite state space $\mathcal{X}=\Gamma \times \mathcal{U}$, where $\Gamma=\{0,\cdots, \maxpacks\}^{{\maxdead}\times K }$ and $\mathcal{U}:=\{1, \cdots, U_{\max}\}$ for a finite $U_{\max}$\footnote{Essentially, $\trafficextra(t)$ assigns labels to $\tau(t)$ to allow more complicated dependencies in $\tau(t)$. If $\mathcal{U}=\varnothing$, then $\tau(t)$ itself evolves as a Markov chain.}.}

% $\tau(t) \in \Gamma$ and $\trafficextra(t) \in \mathcal{I}$

% $\mathcal{I}=\{1, \cdots, L_{max}\}$ 

% the space $\mathcal{X}=\Gamma \times \mathcal{I}$,

%and further $\tau_l(n)\leq \maxdead$, for all $l\in \users, n\geq 1$. 
% denoted with $\mathcal A$ and has a probability distribution so that, when the traffic is viewed as a stochastic process, it is a stationary ergodic process with the following properties:
% \begin{enumerate}
%     \item $a_l(t)\leq \maxpacks, \forall l\in \users, \forall t$
%     \item $\tau_l(n)\leq \maxdead, \forall l\in \users,\forall n\geq 1$
% \end{enumerate}

Note that the arrival and deadline processes do not need to be i.i.d. across times or users. Since the state space $\mathcal{X}$ is finite, \NEW{$\trafficstate(t)$} \MREMOVE{($\tau(t)$)}is a positive recurrent Markov chain~\cite{dynkin2012theory} and the time-average of any bounded function of \NEW{$\trafficstate(t)$ }\MREMOVE{($\tau(t)$)}is well-defined, in particular, the packet arrival rate $\overline{a}_l$, $l\in \users$,
\be \label{eq:arrival rate}
&\lim_{t\to \infty} \frac{1}{t}{\sum_{s=1}^{t} a_l(s)}=\overline{a}_l.
\ee
See Figure~\ref{fig:tightexample} for an example of a Markovian traffic process.
\begin{figure} [t]
\centering
    \includegraphics[width=0.35\textwidth]{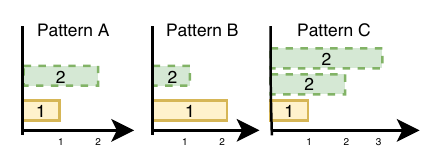}
\caption{An example of a Markovian traffic process with three traffic patterns repeating as $A \to B\to C\to A\cdots$. Each rectangle indicates a packet for a link indicated by its number. The left side of the rectangle corresponds to its arrival time, and its length corresponds to its deadline. For example on pattern $A$, we have 2 packets, 1 from link 2, with deadline 2 slots after the arrival, and 1 from link 1, with deadline in the same slot.
}
\vspace{- 0.1 in}
\label{fig:tightexample}
\end{figure}

\textit{Buffer Dynamics.} The buffer of link $l$ at time $t$ contains the existing packets at link $l$ which have not expired yet and also the newly arrived packets $\tau_l(t)$. 
Formally, we define the buffer of link $l$ by a vector $\Psi_l(t) =(\Psi_{l,d}(t); d=1,\cdots,\maxdead)$, where $\Psi_{l,d}(t)$ is the number of packets in the buffer with \textit{remaining} deadline $d$ at time $t$. The remaining deadline of each packet in the buffer decreases by one at every time slot, until the packet is successfully transmitted or reaches the deadline $0$, which in either case the packet is removed from the buffer, i.e., the buffer at the beginning of slot $t+1$ is
\be
\Psi_{l,d}(t+1)=\Psi_{l,d+1}(t)+\tau_{l,d}(t+1)-I_{l,d+1}(t),
\ee
where $I_l(t)=\sum_{d=1}^{\maxdead} I_{l,d}(t)\leq 1$, and $I_{l,d}(t)=1$ if the scheduler selects a packet with deadline $d$ to transmit at time $t$ on link $l$. By convention, we set  $\Psi_{l,{\maxdead +1}}(t)=0$, $\Psi_{l,0}(t)=0$. We define the network buffer state as $\Psi(t)=(\Psi_l(t); l \in \users)$.  

\textit{Delivery Requirement and Deficit.} As in \cite{houborkum09,houkum09,kumar10,srikant10,kang2014performance}, we assume that there is a minimum delivery ratio $p_l$ (QoS requirement) for each link $l$, $l \in \users$. This means the scheduling algorithm must successfully deliver at least $p_l$ fraction of the incoming packets on each link $l$ in long term. Formally,
\be \label{eq:delivery ratio}
&\liminf_{t\to \infty} \frac{\sum_{s=1}^tI_l(s)}{\sum_{s=1}^ta_l(s)}\geq p_l.
\ee

We define a deficit $w_l(t)$ which measures the amount of service owed to link $l$ up to time $t$ to fulfill its minimum delivery rate. As in \cite{kang2014performance,srikant10}, the deficit evolves as
\be \label{eq:deficit queue}
w_l(t+1)=\Big[w_l(t)+\widetilde{a}_l(t)-I_l(t)\Big]^+,
\ee
where $[\cdot]^+=\max\{\cdot,0\}$, and $\widetilde{a}_l(t)$ indicates the amount of deficit increase due to packet arrivals. For each packet arrival, we should increase the deficit by $p_l$ on average. For example, we can increase the deficit by exactly $p_l$ for each packet arrival to link $l$, or use a coin tossing process as in~\cite{kang2014performance,srikant10}, i.e., each packet arrival at link $l$ increases the deficit by one with the probability $p_l$, and zero otherwise. We refer to $\tilde{a}_l(t)$ as the \textit{deficit arrival} process for link $l$. Note that it holds that
\be \label{eq:deficitarr}
& \lim_{t\to \infty} \frac{1}{t}\sum_{s=1}^{t} {\widetilde{a}_l(s)}=\overline{a}_lp_l:=\lambda_l, \ l\in \users.
\ee
We refer to $\lambda_l$ as the deficit arrival rate for link $l$. We would like to emphasize that the arriving packet is always added to the link's buffer, regardless of whether and how much deficit is added for that packet. Also note that in \dref{eq:deficit queue}, each time a packet is scheduled from the link, $I_l(t)=1$, the deficit is reduced by one. The dynamics in~\dref{eq:deficit queue} define a deficit queueing system, with bounded increments/decrements, whose stability, e.g.,
\be \label{eq:strong stability}
& \limsup_{t\to \infty} \frac{1}{t}\sum_{s=1}^t\mathbb{E}[w_l(s)] <\infty,
\ee
implies that \dref{eq:delivery ratio} holds\footnote{Actually only the rate stability is enough to establish \dref{eq:delivery ratio}~\cite{neely2010queue}, however we consider this stronger notion of stability.}. Define the vector of deficits as $w(t)=(w_l(t), l \in \users)$. The system state at time $t$ is then defined as 
\NEW{$
\statet= (\Psi(t), w(t), \trafficstate(t)).
$} 
%\FIX{I think we might need to include the traffic state here for this to be a Markov chain. So $\statet=(\Psi(t),w(t), \text{trafstate}(t))$. Otherwise the markovian property might not hold, i.e. it could be that $S(t_3)|S(t_2),S(t_1) \not \sim S(t_3)|S(t_2)$ }
%
% A packet tuple or packet is denoted by $(w_l(t),d(t))_l$ where $l$ is the link that the packet belongs to, $w_l(t)$ is the weight of the packet (the deficit length of its link) and $d(t)$ is the remaining deadline of the packet. 

\textit{Objective.} Define $\cpolicy$ to be the set of all causal policies, i.e. policies that do not know the information of future arrivals and deadlines (\NEW{and the hidden state of the traffic process $\trafficstate(t)$}) in order to make scheduling decisions.
For a given traffic process \NEW{$\trafficstate(t)$} \MREMOVE{($\traffic(t)$)}, with fixed $\overline a_l$, defined in \dref{eq:arrival rate}, we are interested in causal policies that can stabilize the deficit queues for the largest set of delivery rate vectors $\ve{p}=(p_l, l \in \users)$, or equivalently largest set of $\bm{\lambda}=(\lambda_l := \overline a_l p_l, l \in \users)$ possible. 
For a given traffic process, we say the rate vector $\bm{\lambda}=(\lambda_l,l\in \mathcal K)$ is supportable under some policy $\mu \in \cpolicy$ if all the deficit queues remain stable. Then one can define the supportable (real-time) rate region of the policy $\mu$ as 
\be
\Lambda_{\mu}=\{\bm{\lambda} \geq 0: \bm{\lambda} \text{ is supportable by } \mu\}. 
\ee
Note that for a given traffic distribution, a vector $\bm{\lambda}$ corresponds to a single vector of delivery rate requirements $\ve{p}$ exactly. The supportable rate region under all the causal policies is defined as $\Lambda = \bigcup_{\mu \in \cpolicy} \Lambda_{\mu}$. The overall performance of a policy $\mu$ is evaluated by the efficiency ratio $\gamma_\mu^\star$ which is defined as 
\be \label{eq:capacity region}
\gamma_\mu^\star = \sup \{\gamma: \gamma \Lambda \subseteq \Lambda_\mu\}.
\ee
For a casual policy $\mu$, we aim to provide a \textit{universal lower bound} on the efficiency ratio that holds for ``all'' Markovian traffic processes (without knowing the transition probability matrix).

\section{Randomized Scheduling Algorithms}
In this section, we present our randomized scheduling algorithms. We start with the collocated networks, and then proceed to general networks.

\subsection{Collocated Networks}\label{sec:col}
In a collocated network, only one of the links can transmit \NEW{a } packet at any time. Hence the interference graph $\ifgraph$ is a complete graph.

Define $e_l(t)=\min\{ d:\Psi_{l,d}(t)>0 \}$ to be the deadline of the earliest-deadline packet available at link $l$ at time $t$. By convention, the minimum of an empty set is considered infinity.
We use a tuple $(w_l(t),e_l(t))_{l}$ to denote the earliest-deadline packet of link $l$ with deadline $e_l(t)$ and link deficit $w_l(t)$.  
% dominates a packet  $B=(w_j(t),d_j(t))_{j}$ at time $t$ if it has more or equal weight but expires sooner or at the same time, i.e., $w_1(t) \geq w_2(t)$, $d_1 \leq d_2$. If one of the inequalities is strict, we call it a strict dominance. A non-dominated packet is a packet that at a given time is not dominated strictly by any other available packet. 
We make the following dominance definition.
\begin{definition}\label{def:dominance}
We say that a link $l_1$ dominates a link $l_2$ at time $t$ if 
$w_{l_1}(t) \geq w_{l_2}(t)$ and $e_{l_1}(t) \leq e_{l_2}(t)$. If one of the two inequalities is strict, we call it a strict dominance. A  non-dominated link is a nonempty link that is not dominated strictly by any other link at that time.
\end{definition}
Recall that $\nonempty$ is the set of links with nonempty buffers. At every time slot, we first find the set of non-dominated links $\nondominated(t)$. One way to do that is as follows: 

\begin{algorithm}
\caption{Finding Set of Non-dominated Links}
\label{alg:nondominated}
\begin{algorithmic}[1]
    \State $H \gets \nonempty$, $\nondominated(t) \gets\varnothing$, $i \gets 0$
    \While{$H \neq \varnothing$}
    \State $i \gets i+1$
    \State Find the largest-deficit non-dominated link $h_i \in H$.
    \State Add ${h_i}$ to  $\nondominated(t)$
    \State Remove ${h_i}$ and all the links dominated by it, i.e. 
    $$
    H \leftarrow H \setminus \{l \in H: e_{l}(t) \geq e_{h_i}(t)\}. 
    $$
    \EndWhile
    \end{algorithmic}
\end{algorithm}

\begin{figure}[t]
\centering
        \includegraphics[width=0.4\textwidth]{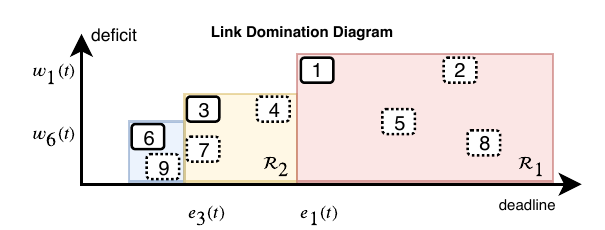}
\caption{An example for non-dominated links. Each numbered rectangle denotes the earliest-deadline packet of a link. A solid rectangle indicates that the link is non-dominated. Dashed rectangles (links) that fall in regions $\mathcal R_i$ will be dominated.}
\label{fig:dominatedpackets}
\vspace{- 0.1 in}
\end{figure}

Algorithm~\ref{alg:nondominated} returns a set $\nondominated(t) = \{h_1,..,h_k\}$, where $h_i$ is the link selected in the $i$-th iteration, and the links are ordered in the order of their deficits, i.e., $w_{h_1}(t)>w_{h_2}(t)>\cdots>w_{h_k}(t)$. See Figure~\ref{fig:dominatedpackets} for an illustrative example of the non-dominated links.
Our scheduling algorithm transmits the earliest-deadline packet of one of the links $h_i\in \nondominated(t)$ randomly, where the probabilities $p_{h_i}(t)$ are computed recursively as in Algorithm~\ref{alg:collocated}.
\begin{algorithm}
\caption{\ALGONAME:\ Randomized Scheduling in Collocated Networks}
\label{alg:collocated}
\begin{algorithmic}[1]
\State Use Algorithm~\ref{alg:nondominated} to find $\nondominated(t)=\{h_1,..,h_k\}$.
\State $r \gets 1$ 
\For{$i=1$ to $k-1$}  
\State $p_{h_i}(t)= \min\Big(1 - \frac {w_{h_{i+1}}(t)} {w_{h_i}(t)},r\Big)$ \label{alg-eq:assign-probs-col}
\State $r \gets r-p_{h_i}(t)$
\EndFor
\State $p_{h_k}(t) = r$
\State Send the earliest-deadline packet from link $h_i$ with probability $p_{h_i}(t)$.
\end{algorithmic}
\end{algorithm}
We refer to Algorithm~\ref{alg:collocated} as {\ALGONAME} which stands for \textit{Adaptive Mixing over Non-Dominated links}.

% \[
% p_{h_i}(t) =
%  \begin{dcases}
% \min\left\{ 
% 1 - \frac {w_{h_{i+1}}(t)} {w_{h_i}(t)}, 
% 1-\sum _{j=1}^{i-1} p_{h_j}(t) 
% \right\} &  i < k \\
%  1-\sum _{j=1}^{i-1} p_{h_j}(t) & i= k
%     \end{dcases}
% \]
% Where one can compute the probabilities by starting from smaller and proceeding to larger values of $i$. We transmit one of $h_i$ with probability $p_{h_i}(t)$ 

\begin{theorem}\label{th:colo}
In a collocated wireless network with $K$ links, $\ALGONAME$ achieves an efficiency ratio of at least 
\be
\gamma_{\ALGONAME}^\star \geq 1- \lp 1-\frac{1}{K}\rp^K > \frac{e-1}{e}.
\ee
\end{theorem}
%Further note that the right hand is bounded below by $\frac{e-1}{e}$
% \TODO{I changed this since its more accurate. In the original paper they have a result that on each slot depends on the number of packets that have positive probability, which in our case will always  be less than the number of users. In their case they don't have an upper bound on the number of non-dominated packets}
%We start assigning probabilities starting from $i=1$ while there is no probability unassigned, increasing $i$:
%\[
%p_{h_i}(t) = 1 - \frac {w_{h_{i+1}}(t)} 
%{w_{h_i}(t)}
%\]
%If all $k$ packets are considered and there is remaining unassigned probability, it is assigned to the last link considered, $h_k$.

\begin{remark} Note that {\ALGONAME} has an efficiency ratio which is bounded below by $0.63$, regardless of the number of links. In contrast, we can construct Markovian traffic processes where the efficiency ratio of LDF is less than $1/2+\epsilon$~\cite{kang2014performance}. For example, for the traffic patterns of Figure \ref{fig:tightexample} in the model section, we will see in simulations in Section~\ref{sec:sim} that, while {\ALGONAME} can achieve delivery ratios close to $0.99$, LDF cannot do better than $0.5+\epsilon$. Note that our traffic model does allow traffic patterns as in Figure \ref{fig:tightexample}, since we do not need the traffic Markov chain to be aperiodic. 
\end{remark}

\subsection{Multipartite Networks and General Networks}
Consider the set of all maximal independent sets $\allMIS$ of the interference graph $\ifgraph$. Our randomized algorithm  selects a maximal independent set  (\MIS) $M \in \allMIS$ probabilistically and schedules the earliest-deadline packets of the induced maximal schedule $M\cap \nonempty$. Recall that $\nonempty$ is the set of links with nonempty buffers. We refer to this algorithm as {\MPALGONAME} which stands for \textit{Adaptive Mixing over Maximal Schedules}.
Before presenting the algorithm, we make a few definitions.
%Let the number of \MIS\ be $\numMIS=|\allMIS|$,
\begin{definition}
The weight of a {\MIS} $M\in \allMIS$ at time $t$ is 
\be \label{eq:ordering}
W_M(t) = \sum_{l\in M \cap \nonempty} w_l(t).
\ee
Let $R=\numbMIS$. We index and order $M\in \allMIS$ such that $M_i$ has the $i$-th largest weight at time $t$, i.e., 
$$W_{M_1}(t) \geq  W_{M_2}(t) \cdots \geq W_{M_R}(t).$$ 
\end{definition}
\begin{definition}
Define the subharmonic average of weights of the first $n$ {\MIS}, $n \leq R$,  at time $t$ to be
\be \label{def:subhrm}
C_n(t) = \frac {n-1}{\sum_{i=1}^{n} (W_{M_i}(t))^{-1}}. 
\ee
\end{definition}
%\begin{definition}
The probabilities used by {\MPALGONAME} to select a \MIS\ $M_i$, at time $t$, are as follows
\be \label{eq:mwisprob}
    p^{\nstar}_{M_i}(t)\equiv p^{\nstar}_{i}(t) = \begin{dcases}
         1- \frac {\subhrm{\nstar}} {\misw{i}} &  1\leq i\leq \nstar \\
        0 & \nstar<i\leq R 
    \end{dcases}
    \ee
where $\nstar$ is the largest $n$ such that $\{p^{n}_i(t), 1 \leq i \leq n\}$ defines a valid probability distribution over $1 \leq i \leq n$.
%i.e., $\nstar:= \nstar(t)= \max\{ n: p^n_i(t)\geq 0, \forall i \leq \nstar \}.$
%\end{definition}
%\subsubsection{Complete Multipartite Networks} The sets in $\allMIS$ constitute a partition of the nodes. 
%Note that the n-candidate-probabilities always sum up to 1, but sometimes they can be negative. Therefore we are seeking the largest number that results in valid probabilities. 
Noting that $p^{n}_{i}(t)\geq p^{n}_{i+1}(t)$ for $i<n$, and $\sum_{i \leq n}p^{n}_{i}(t)=1$, $\nstar$ is therefore given by 
\be \label{def:nstar}
\nstar := \nstar(t) = \max\{ n: p^n_n(t)\geq 0\}.
\ee
We drop the dependence on $t$ for $\nstar(t)$ when there is no ambiguity. Algorithm~\ref{alg:general} gives a description of {\MPALGONAME} where $\nstar$ is found using a binary search. Then {\MPALGONAME} selects a {\MIS} $M_i$ with probability $p^{\nstar}_{i}(t)$ as in \dref{eq:mwisprob}.

%First note that $p_i^{n}\geq p_{i+1}^{n}$ for all valid $i$.
%Therefore we are looking for the first $n$ with $p_n^{n}\geq 0$
%We can do that by starting with $n=\numMIS$. Find the largest \MIS\ $V_i$ with $p^{n}_{i}\leq 0$.

\begin{algorithm}
\caption{\MPALGONAME: Randomized Scheduling in General Interference Graphs}
\label{alg:general}
\begin{algorithmic}[1]
    \State $n_1 \gets 1, n_2 \gets \numbMIS$
    \While{$n_1 \neq n_2$}
        \State $n \gets \ceil*{\frac{n_1+n_2}{2}}$
        \If{$p_{n}^{n}(t) \geq 0$}
            \State $n_1 \gets n$
        \Else 
            \State $n_2 \gets n-1$
        \EndIf
    \EndWhile
    \State $\nstar \gets n_1$
\State Select \MIS\ $M_i$ with probability $p_{M_i}^{\nstar}(t)$ as in \dref{eq:mwisprob} and transmit the earliest-deadline packet  of each link in $M_i$.
\end{algorithmic}
\end{algorithm}
The following theorem states the main result regarding the efficiency ratio of \MPALGONAME.
\begin{theorem}\label{main-theorem-general}
In a wireless network with interference graph $\ifgraph$ and maximal independent sets $\allMIS$, the efficiency ratio of \MPALGONAME\ is at least
\[\gamma_{\MPALGONAME}^\star \geq \frac {\numbMIS}{2 \numbMIS - 1} >\frac{1}{2}.\]
\end{theorem}
A special case of this theorem is for networks with a complete $n$-partite interference graph, $n\geq2$. In a complete $n$-partite graph, with $n$ components, $V_1, \cdots,V_n$, links in each component do not share any edge but there is an edge between any two links in different components. Hence, each component $V_i$, $1\leq i\leq n$ is a {\MIS}. We state the result as the following corollary which immediately follows from Theorem~\ref{main-theorem-general}.  
\begin{corollary}
For a wireless network with a complete $n$-partite interference graph, under \MPALGONAME,
\[\gamma_{\MPALGONAME}^\star \geq \frac {n}{2 n - 1}.\]
\end{corollary}
% \begin{proof}
% Notice that in an $n$-partite graph each component is a maximal independent set. Therefore we have $n$ maximal independent sets and Theorem \ref{main-theorem-general} applies directly.
% \end{proof}
\begin{remark}
We emphasize on the importance of Theorem~\ref{main-theorem-general} using a simple interference graph with `star' topology. This is a special case of a bipartite graph with only two components, $V_1$ is the center node, and $V_2$ are the leaf nodes. Notice that the guarantee of {\MPALGONAME} in this case is at least $\frac{2}{3}$, regardless of the number of nodes $K$. This is a significant improvement over LDF, whose efficiency ratio is at least $\frac{1}{K}$ under i.i.d. traffic but not better than  $\frac{1}{\sqrt{K-1}+1}$ under Markovian traffics~\cite{kang2014performance}.
\end{remark}
\begin{remark} We note that the computational complexity of {\MPALGONAME} could be high for general graphs as it requires finding an ordering of maximal schedules. However, it is easily applicable for $n$-partite graphs or small graphs. Moreover, we can further approximate the algorithm by only ordering a subset of maximal schedules as opposed to finding all of them. The randomization in {\MPALGONAME} can be also potentially implemented in a distributed manner by using distributed CSMA-like schemes such as~\cite{ghaderi2010design, ni2012q, shah2010delay}.   
\end{remark}

\section{Analysis Technique}
We provide an overview of the techniques in our proofs. 
%Given a policy $\mu$, we use $\statet=(\Psi^{\mu}(t), w^{\mu}(t))$ and $I^{\mu}(t)$, $t\geq 0$, to denote the state and schedule processes, respectively. We drop the dependence on $\mu$ when there is no ambiguity. 
We first mention a lemma below which should be intuitive.
\begin{lemma}
Without loss of generality, we consider natural policies that use a maximal schedule to transmit at each time. Further, if a link is included in the schedule, its
earliest-deadline packet will be selected for transmission.  
\end{lemma}
\begin{proof}
%\TODO{Replaced a short phrase "immediate" with the complete argument}
The proof is through exchange arguments.
%The first part is immediate.

 For the first part, assume that a policy $\mu$ at time $t_0$ chooses a non-maximal schedule, hence a packet $x$ from link $l$ could have been included in the schedule. Consider an alternative policy $\mu^\prime$ that does schedule any link that could have been included at time $t_0$ so that the schedule becomes maximal, and for the rest of the time, it transmits exactly the same packets as the initial policy $\mu$, except for the transmission of any packet $x$, if $\mu$ schedules it at a later point. This results in $\sum_{s=1}^{t}I_l^{\mu'}(s)\geq \sum_{s=1}^{t} I_l^{\mu}(s), \forall t\geq 1$, and at the same time every schedule transmitted by $\mu^\prime$ for $t\leq t_0$ is maximal. We can repeat this argument for times $t>t_0$ to convert $\mu$ to a policy $\widetilde \mu$ that transmits maximal schedules. We then have
 $\sum_{s=1}^{t}I_l^{\widetilde \mu}(s)\geq \sum_{s=1}^{t} I_l^{\mu}(s), \forall t\geq 1$
 and from (\ref{eq:delivery ratio}) we see that any delivery ratio supported by $\mu$ is also supported by $\widetilde \mu$.
%%%%%%% END OF FIRST PART 

For the second part, consider a policy $\mu$ that at some time $t_0$ transmits a packet that is not the earliest-deadline packet $x_1 =(w_1(t),d_1)_l$ in link $l$. Then there is some other packet $x_2=(w_1(t),d_2)_l$ in  link $l$ with $d_2<d_1$. If we let $\mu$ transmit $x_2$ instead of $x_1$, the buffer state will be improved since we will have the same set of packets in link $l$ except for one packet with a longer deadline now. Further, the link's deficit will not change. 
\end{proof}
\textit{Frame Construction.} A key step in the analysis of our scheduling algorithms is a careful frame construction. We emphasize that the frame construction is only for the purpose of analysis and is \textit{not} part of our algorithms. The F-framed construction in \cite{kang2014performance} only works for i.i.d. arrivals and deadlines. Here, we need a construction that can handle our Markovian traffic model. We present this construction below where frames have random length as opposed to fixed length in \cite{kang2014performance}. 

\begin{definition}[Frames and Cycles]\label{def:frame}
Starting from an initial \NEW{complete } traffic state $\trafficstate(0)=\ve{x}\in \mathcal{X}$, let $t_i$ denote the $i$-th return time of traffic Markov chain $\trafficstate(t)$ to $\ve{x}$, $i=1,\cdots$. By convention, define $t_0=0$. The $i$-th cycle $\mathcal{C}_i$ is defined from the beginning of time slot $t_{i-1}+1$ until the end of time slot $t_i$, with cycle length $C_i=t_i-t_{i-1}$. Given a fixed $k\in \mathbb N$, we define the $i$-th frame $\mathcal{F}^{(k)}_i$ as $k$ consecutive cycles $\mathcal{C}_{(i-1)k+1}, \cdots, \mathcal{C}_{ik}$, i.e., from the beginning of slot $t_{(i-1)k}+1$ until the end of slot $t_{ik}$. The length of the $i$-th frame is denoted by $F^{(k)}_i=\sum_{j=(i-1)k+1}^{ik}C_j$. Define $\mathcal J(\mathcal{F}^{(k)})$ to be the space of all possible traffic patterns $(\tau(t), t\in \mathcal{F}^{(k)})$ during a frame $\mathcal{F}^{(k)}$. Note that these patterns start after $\ve{x}$ and end with $\ve{x}$.
\end{definition} 
By the strong Markov property and the positive recurrence of traffic Markov chain, frame lengths $F^{(k)}_i$ are i.i.d with mean $\E[F^{(k)}]=k\E[C]$, where $\E[C]$ is the mean cycle length which is a bounded constant~\cite{dynkin2012theory}. In fact, since state space $\mathcal{X}$ is finite, all the moments of $C$ (and $F^{(k)}$) are finite.     
%Probability of seeing a pattern $J=(\tau(t), t=1, \cdots, T)\in \mathcal{J}(\mathcal{F})$ is $\pi(J)=\Prob(F=T)\Prob(J|F=T)$ which is well defined by the positive recurrence. 
We choose a fixed $k$, and, when the context is clear, drop the dependence on $k$ in the notation. 

Define the class of \textit{non-causal $\mathcal{F}$-framed} policies $\strat$ to be the policies that, at the beginning of each frame $\mathcal{F}_i$, have complete information about the traffic pattern in that frame, but have a restriction that they drop the packets that are still in the buffer at the end of the frame. Note that the number of such packets is at most $\maxdead \maxpacks K$, which is negligible compared to the average number of packets in the frame, $\overline a_l\E[F]=\overline a_lk\E[C]$, as $k \to \infty$.  Define the rate region 
\be
&\Lambda_{NC}(\frm) = \bigcup_{\mu \in \strat}\Lambda _{\mu}.
\ee
Given a policy $\mu \in \strat$, the time-average service rate $\bar{I}_l$ of link $l$ is well defined. In fact, by the renewal reward theorem (e.g.~\cite{ross2013applied}, Theorem 5.10), and boundedness of $\E[F]$, 
\be
\lim_{t \to \infty} \frac{\sum_{s=1}^{t}I_l(s)}{t}=\frac{\E\left[\sum_{t\in \mathcal{F}}I_l(t)\right]}{\E[F]}=\bar{I}_l.
\ee
Similarly for the deficit arrival rate $\lambda_l$, defined in \dref{eq:deficitarr}, 
%and from the definition of rate we have convergence in probability, we have, $\forall \epsilon>0, \forall \delta \in (0,1)$ there is $F_{0,l}$ such that for all $F\geq F_{0,l}$
% \[\Pr \lb  \left|\frac {\sumt a_l(t) } {F_l} - \lambda_l \right| < \epsilon_l \rb > \delta \]
% \[\Pr \lb  -\epsilon+\lambda_l< \frac {\sumt a_l(t) } {F}  < \lambda_l+\epsilon \rb \geq  \delta \]
\be \label{eq:lambdarenewal}
\frac{\E[\sumt \widetilde{a}_l(t)]}{\E[F]}=\lambda_l,\ l\in \users.
\ee
In Definition~\ref{def:frame}, each frame consists of $k$ cycles. Using similar arguments as in \cite{kang2014performance}, it is easy to see (and it is intuitive) that 
\[\liminf_{k\to \infty}\Lambda_{NC}(\frm^{(k)}) \supseteq \interior (\Lambda). \]
where $\interior(\cdot)$ is the interior.
Hence, if we prove that for a causal policy ALG, there exists a constant $\rho$, and a large $k_0$, such that for all $k\geq k_0$, 
\be \label{eq:region2}
\rho\interior(\Lambda_{NC}(\frm^{(k)})) \subseteq \Lambda_{ALG},
\ee
then it follows that $\Lambda_{ALG} \supseteq \rho \interior(\Lambda)$. For our algorithms, we find a $\rho$ such that \dref{eq:region2} holds for \textit{any traffic process} under our model. Then it follows that $\gamma^\star_{ALG}\geq \rho$.

We define the \textit{gain} of a policy $\mu$ at time $t$ as
\be
&\mathcal G_{\mu}(t)=\sum_{l \in \users}w^\mu_l(t)I^\mu_l(t),
\ee
and the gain over a frame is $\sumt \mathcal G_{\mu}(t)$. To prove \dref{eq:region2}, we rely on comparing the gain (total deficit of packets transmitted) by ALG and an optimal max-gain non-causal policy over a frame. The following proposition states the result for any general \NEW{interference } graph.
\begin{proposition}
\label{prop} 
Consider a frame $\frm \equiv \frm^{(k)}$, for some fixed $k$ based on returns of traffic process  \NEW{$\trafficstate(t)$ } \MREMOVE{($\traffic(t)$)} to a state $\ve{x}$. Let $\|w(t_0)\|=\suml w_{l}(t_0)$ be the norm of the initial deficit vector at the start of the frame. Suppose for a causal policy {ALG}, given any $\epsilon>0$, there is a $W^\prime$ such that when $\|w(t_0)\|>W^\prime$,
\be
\frac {\E\left[\sumt \gain_{\ALG}(t)  | \statez\right]}{\E\left[\sumt \gain_{\mu^\star}(t) | \statez\right]} \geq \rho-\epsilon,
\ee
where \NEW{$\statez=(\Psi(t_0),w(t_0), \trafficstate(t_0))$}, and $\mu^\star$ is the optimal non-causal policy that maximizes the gain over the frame. Then for any $\lambda \in \rho \interior( \Lambda_{NC}(\frm))$, the network state process  $\{\statet\}$ is positive recurrent, and further, the deficit queues are bounded in the sense of \dref{eq:strong stability}. \end{proposition}

The proof of Proposition~\ref{prop} is provided in Section~\ref{sec:prop_proof}.

\textit{Gain Analysis.} With Proposition~\ref{prop} in hand, we analyze the achievable gain of our algorithm over a frame, compared with that of the optimal non-causal policy $\mu^\star$. Since characterizing $\mu^\star$ is hard, we extend a coupling technique from~\cite{chin2006online, jez2011one, bienkowski2011randomized,jez2012online} (developed for constant-weight single buffer analysis)  to stochastic process \NEW{$(\Psi(t), w(t),\trafficstate(t))$ } in a general network. 

Consider a state $(\Psi(t), w(t),\trafficstate(t))$ under our randomized algorithms at time $t \in \frm$, and the state $(\Psi^{\mu^\star}(t), w^{\mu^\star}(t), \trafficstate(t))$ under the optimal policy $\mu^\star$. \NEW{Of course, the traffic process $\trafficstate(t)$ is the same for the entire time in the frame for both algorithms. }
 We change the state of $\mu^\star$ (by modifying its buffers and deficits) to make it identical to $(\Psi(t), w(t),\trafficstate(t))$, but also give $\mu^\star$ a larger gain $\again_{\mu^\star}(t)>\gain_{\mu^\star}(t)$ that can ensure the change is advantageous for $\mu^\star$ considering the rest of the frame.  
% extra gain $\Delta_t$ to ensure the change is advantageous for $\mu^\star$ considering the rest of the frame, i.e., 
% \be\label{eq:advantagous}
% \Delta_t+\sum_{s\in \frm:s> t}\gain_{\mu^\star}(t)|_{\state^\mu(t)}>\sum_{s\in \frm:s> \gain_{\mu^\star}(t)|_{\state^{\mu^\star}(t)},% \ee
%here $\cdot|_{\state}$ indicates the starting statettin$\again_{\mu^\star}(t)=\gain_{\mu^\star}(t)+\Delta_t$
Then, taking the expectation $\E[\gain^\prime(t)]$ with respect to the random decisions of our algorithm, {\ALGONAME} or {\MPALGONAME}, and traffic patterns in a frame, we can bound the optimal gain of $\mu^\star$. Then we can prove the main results in view of Proposition~\ref{prop}.      

The gain analysis of {\ALGONAME} in collocated networks and {\MPALGONAME} in general networks is presented in Sections~\ref{sec:gain_collocated} and \ref{sec:gain_general}, respectively.
\section{Proofs of Main Results}
We first provide the proof of Proposition~\ref{prop} and then provide the gain analysis of our algorithms. In what follows, we define 
\be \label{eq:maxlink}
w_{max}(t)=\max_{l \in \users} w_l(t)\mathds{1}(\Psi_l \neq 0),
\ee
to be the maximum deficit of a nonempty link at time $t$. Also define $[N]:=\{1,2,...,N\}$. We use $\E_{X}[\cdot]$ to denote conditional expectation $\E[\cdot |X]$. $\E^{Y}[\cdot]$ is used to explicitly indicate that expectation is taken with respect to some random variable $Y$. %Denote $x^-=\lim_{\epsilon \uparrow 0} (x+\epsilon)$. 
$|A|$ is used to denote the cardinality of set $A$.

\subsection{Proof of Proposition~\ref{prop}}\label{sec:prop_proof}
%\begin{proof}
We look at the state process $\{\statet\}$ at times $t_i$ when frames start. We show that this sampled chain is positive recurrent and further its mean deficit size is stable in the sense of~\dref{eq:strong stability}. From this it follows that the original process $\{\statet\}$ is also stable as the mean frame size $\E[F]$ is bounded and the mean deficits within a frame can change at most by $a_{max}K\E[F]$.   

%Consider an initial state $\statez$ and the next sampled state $\state(t_0+F)$.
Since $\lambda \in \rho \interior(\Lambda_{NC})$, we have for some $\epsilon>0$, and some policy $\mu \in \strat$,
\be
\lambda \E[F](1+2 \epsilon) \preceq \rho \E[\sumt I^{\mu}(t)],  \label{eq:lambdabound}
\ee
% \begin{align}
% \lambda F(1+\delta) \preccurlyeq \rho' \sumj \pi(J) \xi_J  \Rightarrow \nonumber \\
% w_l(t_0) \lambda_l F ( 1+\delta) \leq \rho' w_l(t_0) \sumj \pi(J) \xi_{J,l}, \forall l \in \users \Rightarrow \nonumber \\
% \suml w_l(t_0) \lambda_l F ( 1+\delta) \leq \rho' \suml w_l(t_0) \sumj \pi(J) \xi_{J,l} \label{eq:lambdabound}
% \end{align}
where $\preceq$ is the component-wise inequality between vectors. This is simply due to the fact that in each frame, the number of deficit arrivals $\sumt \widetilde{a}(t)$ and the number of departures under the policy $\mu$ are i.i.d across the frames, with means $\E[F]\lambda$ and $\E[\sumt I^{\mu}(t)]$, respectively, by the renewal reward theorem. Hence, to ensure stability,~\dref{eq:lambdabound} must hold. Next, consider the Lyapunov function
\ben
V(t) := V(\statet) = \frac 1 2 \sum_{l \in \users} w_l^2(t).
\een
Let $\{I(t), t\in \frm\}$ denote the scheduling decisions by {\ALG} within the frame. Using~\dref{eq:deficit queue}, we get
\begin{flalign*}
&w_l^2(t+1)-w_l^2(t) \leq \left(w_l(t)+\widetilde{a}_l(t)-I_l(t)\right)^2 - w_l^2(t) \\
&= 2 w_l(t) (\widetilde{a}_l(t)-I_l(t)) + (\widetilde{a}_l(t)-I_l(t))^2\\
&\leq 2 w_l(t) (\widetilde{a}_l(t)-I_l(t)) + a^2_{max}.
\end{flalign*}
\noindent
Then we compute the drift over $F$ slots
\begin{flalign}
&V(t_0+F)-V(t_0) = \frac 1 2 \suml \lp w_l^2(t_0+F) - w_l^2(t_0) \rp  \nonumber \\
&= \frac 1 2 \sumt \suml \lp w_l^2(t+1) - w_l^2(t) \rp  \nonumber \\ 
&\leq 
{K a^2_{max} F}/ 2 + \sumt \suml w_l(t) \left(\widetilde{a}_l(t)- I_l(t)\right) \label{eq:driftright}.
\end{flalign}
%where $C_1=\frac {K a^2_{max} F} 2$. 
Let $\E_{t_0}[\cdot]=\E[\cdot|\statez]$. Then, over a frame,
\begin{flalign}
&\E_{t_0}\lb V(t_0+F)-V(t_0) \rb \leq \nonumber\\
% &\E_{t_0}\lb\sumt \suml w_l(t) \widetilde{a}_l(t)-
% \sumt \suml w_l(t) I_l(t)\rb \\
&
\E_{t_0} [\sumt \suml 
w_l(t)  \widetilde{a}_l(t)] -
\E_{t_0}[ \sumt \suml  w_l(t) I_l(t)]
+C_1,
\label{eq:usebounds}
\end{flalign}
where $C_1=K a^2_{max} \E[F]/ 2$.
Noting that
\begin{equation} \label{ineqweight}
w_{l}(t_0) -F\leq w_{l}(t) \leq w_l(t_0) + a_{max}F,
\end{equation}
at any $t\in \frm$, we can bound
\be \label{eq:firsttermintermediate}
\E_{t_0} [\sumt \suml 
w_l(t)  \widetilde{a}_l(t)] 
%&\leq & \E[\suml\sumt w_l(t_0) \widetilde{a}_l(t)] \nonumber \\
%&+&\E[a_{max}F\suml\sumt \widetilde{a}_l(t)] \nonumber \\
%& 
\leq \suml (w_l(t_0) \lambda_l\E[F]) + C_2,
\ee
where we have used \dref{eq:lambdarenewal} and \dref{ineqweight}, and  
% converges in probability to $\lambda_l$ and is dominated for all $F$ by the constant $a_{max}$. Therefore we get by the dominated convergence theorem:
% \[\lim_{F\rightarrow \infty} \frac {\EJ_{t_0}[\sumt a_l(t)]}{F} = \lambda_l\]
% Therefore for all $\theta_l>0$ For sufficiently large $F$ we have $\EJ_{t_0}[\sumt a_l(t)]/F \leq \lambda_l + \theta_l$. 
% \begin{align*}
% \suml (w_l(t_0)+\Delta) F \frac {   \EJ_{t_0} [\sumt a_l(t)]} F  &\leq \\
% \suml (w_l(t_0)+\Delta) F (\lambda_l+\theta_l) 
% \end{align*}
% Now choose $\theta_l$ with $\theta_l = \lambda_l \delta  - \lambda_l \sigma$ for some $\sigma$ with $0<\sigma<\delta$ yielding $\theta_l>0$. There is some $F_{1,l}$ such that the bound holds for all $F\geq F_{1,l}$. Therefore if $F\geq \max_{l\in \users}\{F_{1,l}\}$ we we can further advance (\ref{eq:firsttermintermediate}):
% \begin{align*}
% \suml (w_l(t_0)+\Delta) F  \frac {   \EJ_{t_0} [\sumt a_l(t)]} F  &\leq \\
% \suml (w_l(t_0)+\Delta) F  \lambda_l (1+\delta - \sigma )  &\leq \\
% \suml w_l(t_0)F \lambda_l (1+\delta - \sigma ) + C  \\
% \end{align*}
%Where $\suml \Delta F \lambda_l ( 1+\delta-\sigma)$ is bounded by some constant $C$. 
$C_2=a^2_{max}\E[F^2]K<\infty$. 
%Define the norm of the state as
% \[
% \|\statet\| = \sum_{l\in \users} \sum_{d =1}^{\maxdead} \Psi_l^{d}(t)+
% \sum_{l \in \users} w_l(t)
% \]
% Note that $\|\statez\| \geq S$ implies $W_0 = \suml w_l(t_0) \geq S- K a_{max} \maxdead$.
%For the chosen $\epsilon$, there is some $S'$ such that if $||\statez|| \geq S' \Rightarrow W_0 \geq W'$. Then
%Let $\mu^\star$ be the optimal non-causal policy that maximizes the gain over the frame.
Let $I^\star(t)$ be the scheduling decisions by the policy $\mu^\star$, and $I^\mu(t)$ be the scheduling decisions by the policy $\mu \in \strat$ in~\dref{eq:lambdabound}. 
Note that $\mu^\star$ is the optimal non-causal policy that maximizes the gain over the frame and can transmit packets from a previous frame (included in the initial buffer $\Psi(t_0)$). This only improves the performance of $\mu^\star$, compared to starting with empty buffers, hence,
% We have proved that in the case when an non-constrained Max-Gain policy $\mu$ can transmit packets from a previous frame (packets included in the initial state $\statez$) that we have
% \[\ERJ_{t_0} \lb \sumt \gain_{\ALGONAME}(t) \rb \geq
% (\rho-\epsilon) \EJ_{t_0} \lb \sumt \gain_{\ADV}(t) \rb 
% \]
%But we have further that 
\be \label{eq101}
\E_{t_0}\Big[\sumt \suml w_l^\star(t) I_l^\star(t)\Big] \geq \E_{t_0} \Big[\suml \sumt  w^\mu_l(t) I^\mu_l(t)\Big].
\ee
Using \dref{eq101} and the proposition assumption, given $\epsilon>0$, there is a $W^\prime$ such that, if $\|w(t_0)\|>W^\prime$,
\begin{align}
&\E_{t_0}\Big[\sumt \suml w_l(t) I_l(t)\Big] \geq (\rho-\epsilon) \E_{t_0}\Big[\sumt \suml w_l^\star(t) I_l^\star(t)\Big] \nonumber\\
%& \geq (\rho-\epsilon)\E_{t_0} \Big[ \suml \sumt w_l^\star(t) I_l^\star(t)\Big] \nonumber\\
&\geq (\rho-\epsilon) \E_{t_0} \Big[\suml \sumt  w^\mu_l(t) I^\mu_l(t)\Big] \nonumber\\
&    \geq (\rho-\epsilon)\E_{t_0} \Big[\suml \sumt (w_l(t_0) -F) I^\mu_l(t) \Big] \nonumber \\
&    \geq (\rho-\epsilon)\E_{t_0} \Big[\suml \sumt w_l(t_0) I^\mu_l(t)\Big] - C_3 \label{eq:optleft},
    % &= \sum_{\til I} a_{\til I,J} \suml w_l(t_0) \sumt \til I_l(t,J) - C_2 \nonumber \\
    % &= \suml w_l(t_0)  \sumt \lp \sum_{\til I} a_{\til I,J} \til I_l(t,J) \rp - C_2 \nonumber \\
    % &=  \suml  w_l(t_0) \xi_{J,l} - C_2 \label{eq:optright}
\end{align}
% \begin{equation} \label{eq:optleft}
% \forall J: \suml \sumt w_l^\star(t,J) I_l^\star(t,J) \geq \suml \sumt \widetilde w_l(t,J) \widetilde I_l(t,J)
% \end{equation}
% \begin{align}
%     &\geq \suml \sumt (w_l(t_0) -F) \widetilde I_l(t,J) \nonumber \\
%     &\geq \suml \sumt w_l(t_0) \widetilde I_l(t,J) - C_2 \nonumber \\
%     &= \sum_{\til I} a_{\til I,J} \suml w_l(t_0) \sumt \til I_l(t,J) - C_2 \nonumber \\
%     &= \suml w_l(t_0)  \sumt \lp \sum_{\til I} a_{\til I,J} \til I_l(t,J) \rp - C_2 \nonumber \\
%     &=  \suml  w_l(t_0) \xi_{J,l} - C_2 \label{eq:optright}
% \end{align}
where $C_3=K\E[F^2]$ is a constant. Using \dref{eq:optleft}, \dref{eq:firsttermintermediate}, \dref{eq:usebounds},
% $a_{\til I,J}$ the coefficients of a convex combination over all strategies in $\Lambda_{NC}$ such that
% \[\xi_{J,l} = \sumt \lp \sum_{\widetilde I} a_{\widetilde I,J} \widetilde I_l(t,J) \rp\]
% we can interpret $\xi_{J,l}$ as the expected number of serviced packets when pattern $J$ occurs, for user $l$, over a frame by using the scheduling choices of $\widetilde I$ with probability $a_{\widetilde I, J}$. 
% \DETAILS{The existence of such a convex combination stems from the fact that $\xi_J \in \mathcal{CH}(M_J)$}.
% Take expectations over the pattern $J$
% between the left hand side of (\ref{eq:optleft}) and (\ref{eq:optright}), to get 
% \begin{equation}
%     \E_{t_0}^{J}[\sumt \suml w_l^\star(t,J) I^\star_l (t,J) ] \geq \suml w_l(t_0) \sumj \pi(J) \xi_{J,l}  - C_2
%     \label{eq:boundopt}
% \end{equation}
% where we used the definition of expectation for the right hand side, and we exchanged the order of the summations, factoring $w_l(t_0)$ out since it does not depend on $J$. 
% Now taking expectation between (\ref{eq:driftleft}) and (\ref{eq:driftright}) given the initial state we get
\begin{align}
&    \E_{t_0}\Big[(t_0+F)-V(t_0)\Big] \nonumber \\
\leq&    C_4+ \suml \E[F] w_l(t_0) \lambda_l  -  (\rho-\epsilon)   \suml w_l(t_0)  \E_{t_0}\Big[\sumt I^\mu_l(t)\Big]  \nonumber \\
\leq& C_4 +  \suml w_l(t_0) \lp \lambda_l \E[F]-(\rho-\epsilon)\E_{t_0}\Big[\sumt I^\mu_l(t)\Big] \rp \nonumber\\
\leq& C_4 -\epsilon \E[F]\suml \lambda_l w_l(t_0)  \label{eq:bound1},
\end{align}
% \begin{align}
% &    \E_{t_0}[V(t_0+F)-V(t_0)] \nonumber \\
% \leq&    C_1 + \E_{t_0}[
%     \sumt \suml w_l(t_0) \lp a_l(t)- I_l(t)\rp
%     ] \nonumber \\ 
%  \leq&   C_3 + \suml F w_l(t_0) \lambda_l ( 1+ \delta - \sigma) - 
%     \sumt \suml  \E_{t_0}[w_l(t) I_l(t)]   \label{eq:bound1} \\
% \leq& C_3 +  \suml F w_l(t_0) \lambda_l (1+\delta-\sigma)   \label{eq:bound2} \\ 
%     &- \rho' \lp \suml w_l(t_0) \sumj \pi(J) \xi_{J,l}  - C_2 \rp  \nonumber
%  \\
% \leq& C_4 -\sigma \suml F \lambda_l w_l(t) 
% \end{align}
where $C_4= C_1 + C_2+C_3$,
% and $C_4 = C_3 + \rho' C_2 $. For (\ref{eq:bound1}) we used (\ref{eq:optnonopt}) and then (\ref{eq:boundopt}). 
and in the last inequality we have used (\ref{eq:lambdabound}).
Hence, given any $\delta>0$, $\E_{t_0}\lb V(t_0+F)-V(t_0) \rb \leq - \delta$ if $$\|w(t_0)\|\geq \max \lp (C_4+\delta)/(\epsilon \E[F]\lambda_{min}), W^\prime \rp,$$ 
where $\lambda_{min}=\min_l \lambda_l$. This proves that the network Markov chain is positive recurrent by the Foster-Lyapunov Theorem and further the stability in the mean sense~\dref{eq:strong stability} follows~\cite{meyn1992stability} (note that
the component $\Psi(t)$ lives in a finite state space).
%\end{proof}

%Next, we present the gain analysis for {\ALGONAME} and {\MPALGONAME}, separately.
\subsection{Gain Analysis of {\ALGONAME} in Collocated Networks} \label{sec:gain_collocated}
%In this section, we present the gain analysis of {\ALGONAME} over a frame. Initially, we 
Consider a subclass $\ndpol$ of all the policies that schedule Non-Dominated (ND) links at each slot (recall Definition~\ref{def:dominance}). We refer to policies in $\ndpol$ as \textit{ND-policies}. We show that the optimal ND-policy is close to the optimal non-restricted policy as stated below.
%\begin{lemma*} \TODO{Possibly replace Lemma 4 with this lemma. I want to use Lemma 3 here to possibly simplify lemma 4 (by using the fact that when we schedule a link we know exactly what packet will be scheduled since ADV will schedule the earliest deadline packet wlog)} 
%\end{lemma*}
%\begin{proof}
%Assume that the first time at which $\ADV$ and $\ADVR$ disagree in a \NEW{link} they transmit from is time $t_0$. Let the link chosen by $\ADV$ be $y$ and consider some other link $x$ that strictly dominates link $y$, 
% i.e. $w_x(t_0) \geq w_y(t_0), e_x(t_0) \leq e_y(t_0)$. \TODO{I will complete this after I finish more important parts since its optional}
%\end{proof}

\begin{lemma}  \label{lemma:baselemma}
Consider any policy $\ADV$ for scheduling packets in a frame $\frm$. Then there is an ND-policy $\ADVR \in \ndpol$ such that, under the same pattern $J \in \mathcal J(\frm)$ and initial state $\statez$, %the gain of $\ADVR$ over the frame $\mathcal F$ is such that:
$$
\sumt \gain_{\ADVR}(t) \geq
\sumt \gain_{\ADV}(t) -  a_{max}F^2
$$
where $F$ is the length of the frame.  
\end{lemma}
\begin{proof}
 Suppose the first time $\ADV$ does not schedule a non-dominated link is $t_0$. Suppose $\ADV$ sends earliest-deadline packet $(w_y(t_0),d_y)$ from link $y$ and $(w_x(t_0),d_x)$ be the earliest-deadline packet at a link $x$ ($x\neq y$) that strictly dominates $y$, i.e. $w_x(t_0) \geq w_y(t_0)$, $d_x \leq d_y$. Consider some alternative policy $\ADV'$ which has the same transmissions as $\ADV$ up to time $t_0$ but transmits the packet of $x$ at time $t_0$ instead. Let $w_l'(t)$, $ l\in \users$ denote the link deficits under $\ADV'$. Note that $w'_l(t) = w_l(t), \ \forall t\leq t_0$.
We differentiate between 2 cases:

\begin{enumerate}[leftmargin=*]
    \item $\ADV$ does not transmit packet $x$ in the remaining time slots. In this case, let $\ADV'$ transmit the same packets as $\ADV$ in the remaining slots (after $t_0$).
    Let $I_l(t_1,t_2)=\sum_{t=t_1}^{t_2}I_l(t)$ be the number of packets transmitted between  $t_1$ and $t_2$ at link $l$ under $\ADV$ (and subsequently under $\ADV'$). And let $\Delta \gain:=\sumt \gain_{\ADV'}(t) - \sumt \gain_{\ADV}(t)$.  Then we have
    \begin{flalign*}
    &\Delta \gain \overset{(a)}=w_x(t_0)+I_y(t_0+1,F) - \lp w_y(t_0)+I_x(t_0+1,F)  \rp  \\ 
    &\overset{(b)} \geq w_x(t_0) - w_y(t_0) -F \geq -F
\end{flalign*} 
    To see $(a)$, notice that as a result of transmitting from link $x$ instead of link $y$, the deficit of link $y$ under $\ADV'$ will be one more than that under $\ADV$ at any time $t >t_0$. Similarly, the deficit of link $x$ under $\ADV'$ will be one less than that under $\ADV$ at any time $t >t_0$. In $(b)$, we have used the fact that $I_l(t)\in \{0,1\}$ and $w_x(t_0)\geq w_y(t_0)$.
    
    \item $\ADV$ transmits packet $x$ at some time slot $t_a$ where $t_0<t_a<t_0+d_x$. In this case we let $\ADV'$ transmit the same packets as $\ADV$ for all $t>t_0$ except for time slot $t_a$ in which it transmits packet $y$ instead, which still has not expired yet by the domination inequality $d_y \geq d_x$. 
    % This construction is indeed feasible.  In this case $\ADV'$ and $\ADV$ will have the same deficits after $t_a$. Thus the gain difference will depend on the transmitted packets from links $x$ and $y$ during the time window $[t_0,t_a]$. Links $x$ and $y$ will send the same number of packets in time slots $t_0+1,\cdots,t_a-1$ under $\ADV$ and $\ADV'$ but differ at times $t_0$ and $t_a$. 
     It is easy to check that
     \be
     &\sumt \gain_{\ADV'}(t) - \sumt \gain_{\ADV}(t) = \nonumber \\  
     &w_x(t_0)+w_y'(t_a) +I_y(t_0+1,t_a-1) \nonumber \\
     &- w_y(t_0)-w_x(t_a)-I_x(t_0+1,t_a-1) \label{eq:gaindif}
     \ee
 The total deficit arrival to a link in the frame cannot be more than $a_{\max}F$. Hence,
     \begin{align*}
     w_x(t_a)  
     %w_x(t_0)+\text{(deficit change of $u_x$)} \\
      %      &= w_x(t_0)+\text{(arrivals $u_x$)} - \text{(transmissions $u_x$)} \\
           &\leq w_x(t_0) + a_{\max}F - I_x(t_0,t_a-1)\\
     w_y'(t_a) &\geq w_y(t_0) - I_y(t_0,t_a-1)       
     \end{align*}
     % By similar argument
     % \begin{align*}
     % w_y'(t_a) &= w_y(t_0) + \text{(deficit change of $u_y$)} \\
     %  &= w_y(t_0) + \text{(arrivals of $u_y$)} - P_y(t_0,t_a-1) \\
     %  &\geq w_y(t_0) - P_y(t_0,t_a-1)
     % \end{align*}
   Using these two inequalities in \dref{eq:gaindif} yields
    \be \label{eq:samplegain}
    & \sumt \gain_{\ADV'}(t) - \sumt \gain_{\ADV}(t) \geq 
    -a_{max}F.
    \ee
\end{enumerate}
%Comparing the two cases, we see that inequality~\dref{eq:samplegain} holds in either case.  
    By repeating this process (at most $F$ times), we can transform $\ADV$ to $\ADVR$. From this, the final result follows. 
\end{proof}

\begin{lemma}\label{lemma:boundsadvalg}
For each slot $t\in \frm$, the gain obtained by \ALGONAME, and the amortized gain by any ND-policy $\ADVR$, starting from some state $\statet$ satisfy:
\be
\ER[\again_{\ADVR}(t)|\statet] &\leq& w_{max}(t) + \mathcal E_0 \label{eq200}\\
 \ER[\gain_{\ALGONAME}(t)|\statet] &\geq& w_{max}(t) \rho \label{eq300}
\ee
where  $\rho={\lp 1- \lp 1-\frac{1}{K}\rp^K \rp}$ and $\mathcal{E}_0= a_{max}\maxdead + 2F$, and $\ER[\cdot]$ is expectation with respect to the random decisions of {\ALGONAME}. 
\end{lemma}
\begin{proof}
At time $t$, after the new arrivals have happened, we have state $\statet$. \ALGONAME\ decides probabilistically to transmit a packet $(w_f,e_f)$ from a non-dominated link $f \in \nondominated(t)$, and the ND-policy $\ADVR$ transmits a packet $(w_z,e_z)$ from some other link $z$ . We distinguish two cases following the same method as in \cite{jez2011one} but for time-varying weights.

\begin{enumerate}[leftmargin=*]
    \item $e_f \leq e_z, w_f \leq w_z$:
%    In this case, $\ADVR$ sends a heavier packet that expires later. 
To maintain the same buffers for both algorithms, we remove the packet $e_f$ from the buffer of link $f$ under $\ADVR$ and inject the packet with deadline $e_z$ to link $z$ so that $\ADVR$ gets {a packet with higher deadline and higher weight at the time $t$. 
%This is not an advantageous change only in the case when future arrivals cause the weight of packet $f$ to become higher than that of packet $z$. 
%But 
Since both packets will expire in at most $\maxdead$ slots, the deficit of $f$ can only increase by at most $\maxdead \maxpacks$ before packet $e_f$ expires. Therefore giving $\ADVR$ this additional compensation will guarantee that the modification is advantageous}.  Further, we decrease the deficit from link $f$ by one ($w_f-1$ in $\ADVR$) and we increase the deficit of link $z$ by one ($w_z+1$ in $\ADVR$). Then $\ADVR$ and {\ALGONAME} have the same exact state.
    Making this change in the deficit will reduce the gain for each packet transmitted from link $f$ in the future  by one. To compensate for this, we give $\ADVR$ extra gain which is the number of packets transmitted from link $f$ for the rest of the frame, which is less than $F$.     Hence, the total compensation is bounded by   $F+a_{max}\maxdead.$
% \[ 
% \text{\# packets transmitted by $\ADV$ in the future from } f
% \leq 
% \]
% \[
% \text{(\# packets currently in $f$)}
% + \text{(\# future arrivals in $f$)} 
% \leq 
%     \]
%     \[
%     a_{max} \tau_{max} + F a_{max} =  \mathcal E/2
%     \]
    \item $e_z \leq e_f, w_z \leq w_f$: In this case, we allow $\ADVR$ to additionally transmit the packet $e_f$ at time $t$, and inject a copy of packet $e_z$ to the buffer of link $z$. This makes the buffers identical, but results in the decrease of deficit of link $f$ by one, which might not be advantageous for $\ADVR$ for future times.  
    %$\ADVR$ might strategically not transmit some packet so its deficit increases and collects more weight in subsequent slots. 
    To guarantee that the change is advantageous for $\ADVR$, we give it one extra reward for each possible transmission from link $f$ in the rest of the frame, which is less than $F$. 
\end{enumerate}

Let ${\again_{\ADVR}}^{(h_i)}(t)$ denote the reward (including the compensation) gained by  $\ADVR$ when it transmits a non-dominated packet $h_i$ (recall $h_i$ from Algorithm \ref{alg:nondominated}). Then 
\begin{align}
\ER [{\again_{\ADVR}}^{(h_i)}(t)|\state^t] 
=& \sum_{h_j: j<i}p_{h_j}(t) \lp w_{h_j}(t) + F \rp \nonumber \\
&+w_{h_i}(t) + F+a_{max}\maxdead  \nonumber \\
\leq &w_{h_i}(t) + \sum_{h_j: j<i}p_{h_j}(t) w_{h_j}(t)  +\mathcal E_0 \label{eq:bound-adv-collnet}
\end{align}
where $\mathcal E_0=a_{max}\maxdead +2F$. Using the assigned probabilities (line \ref{alg-eq:assign-probs-col} in Algorithm \ref{alg:collocated}), it is easy to verify that \dref{eq:bound-adv-collnet} attains its maximum for $i=1$, which is equal to $w_{h_1}(t)+\mathcal E_0=w_{max}(t)+\mathcal E_0$. Hence, \dref{eq200} indeed holds.
% Where the two first terms of the right hand side of (\ref{eq:bound-adv-collnet}) are upper bounded by $w_{max}(t)$ (proved in \cite{RAND} and included here for completeness). To see that, consider two non-dominated links selected successively in the algorithm \TODO{ref-alg}, $h_i,h_{i+1} \in \nondominated(t)$ (in the nontrivial case where $\nondominated(t)>1$). From the assigned probabilities (line \ref{alg-eq:assign-probs-col} in algorithm \ref{alg:collocated}) it can be easily verified that
% \begin{gather*}
% w_{h_{i+1}}(t) + \sum_{h_j: j<i+1}p_i(t) w_i(t)
% \leq w_{h_i}(t) + \sum_{h_j: j<i}p_i(t) w_i(t)
% \end{gather*}
% Therefore the maximum value in  (\ref{eq:bound-adv-collnet}) is achieved for $h_1$ yielding the final upper bound regardless of the choice of $\ADVR$ (note $w_{h_1}(t)=w_{max}(t)$)
% \ben
% \ER[\again_{\ADVR}(t)|\state^t]  \leq w_{max}(t) + \extracompcol.
% \een

Now regarding {\ALGONAME}, similar derivation applies as in \cite{jez2013universal} to get the final bound. To see that, first let the number of links with positive probability be $B\leq K$. Then
\begin{flalign*}
&\ER[\gain_{\ALGONAME}(t)|\state^t]  =
\sum_{i\in [B]} w_{h_i}(t)p_{h_i}(t)= \\
& \sum_{i\in [B-1]} w_{h_i}(t)p_{h_i}(t) +
\Big(1- \sum_{i \in [B-1]} p_{h_i}(t) \Big) w_{h_B}(t) \overset{(a)}= \\
&w_{h_1}(t) \Big( 1- \prod _{i=1}^{B-1}(1-p_{h_i}(t) ) \sum_{i=1}^{B-1}p_{h_i}(t)\Big) \overset{(b)} \geq \\
&w_{h_1}(t) \Big(1-\big(\frac{B-1} B\big)^B\Big), 
\end{flalign*}
where $(a)$ follows from the form of probabilities, 
%recursive relation between $w_{h_i}(t)$ and $w_{h_1}(t)$
and $(b)$ follows by applying the inequality between arithmetic and geometric means of $B$ terms: $(1-p_{h_i}(t)),\ i\in[B-1]$, and $\sum_{i=1}^{B-1}p_{h_i}(t)$.
%giving 
%$ \prod _{i=1}^{B-1}(1-p_{h_i}(t) ) \sum_{i=1}^{B-1}p_{h_i}(t) \leq (\frac{B-1} B)^B =  (1- \frac 1 B)^B$
%. For more details we refer the reader to \cite{RAND}.
\end{proof}
% \noindent \textit{Remark:} We emphasize the fact that the amortized gain \TODO{notion of amortized gain mentioned here} of the maximum-gain ND-policy in a frame is a random variable that depends on the random decisions of our algorithm. On the other hand, for a fixed  F-pattern $J$ the actual gain of the maximum-gain ND-policy in the frame will be deterministic and can be precomputed at the beginning of the frame. 

\begin{lemma} \label{lemmas0} 
Over any frame $\frm$, with initial state $\statez=(\Psi(t_0),w(t_0),\trafficstate(t_0))$, and  any ND-policy $\ADVR$. 
%Let $\|w(t_0)\|=\suml w_{l}(t_0)$ be the norm of the initial deficit vector. 
\be 
\lim_{\|w(t_0)\|\to \infty} \frac {\ERJ[\sumt \gain_{\ALGONAME}(t)  | \statez]}{\EJ[\sumt \gain_{\ADVR}(t) | \statez]} \geq \NEW{\rho}
\ee
% the total gain obtained by \ALGONAME, $\ALG$ in a frame and the total gain obtained by the max-gain ND-Adversary $\ADVR$, have the property: for every $\epsilon>0$, there is an $F_0$ such that for every $F$ with $F\geq F_0$ there is $W'$ such that for every initial state $\statez$ with $W_0=\suml w_{l}(t_0)\geq W'$  the following holds 
% \[ \frac {\ERJ[\sumt \gain_{\ALG}(t)  | \statez]}{\EJ[\sumt \gain_{\ADVR}(t) | \statez]} \geq \rho-\epsilon\]
 \end{lemma}
\begin{proof}
Given the initial state $\statez$ and frame size $F$, consider all the traffic patterns of length $F$. Taking expectations of the result of Lemma \ref{lemma:boundsadvalg}, with respect to random traffic patterns $J$ of length $F$, we get 
\NEW{
\[ \ERJ[\ER[\again_{\ADVR}(t)|\statet]| \statez,F]
\leq \ERJ[w_{max}(t) | \statez,F] + {\mathcal E_0}\]
\[
 \ERJ[\ER[\gain_{\ALG}(t)|\statet]| \statez,F]
\geq \ERJ[w_{max}(t)|\statez,F] \rho
\]
}
where \ALG\ = {\ALGONAME}. Now notice that
\NEW{
\begin{align*}
&\ERJ[\ER[\again_{\ADVR}(t)|\statet]| \statez,F]=\\
& \ERJ[\ER[\again_{\ADVR}(t)|\statet, \statez]| \statez,F]=  \ERJ[\again_{\ADVR}(t)|\statez,F]
\end{align*}
}
where the first equality is due to the fact that, given $\statet$ and $F$, the gain of $\ADVR$ at time $t$ depends on current $\statet$ and future traffic pattern in the frame, but not on the past. The second equality is by the tower property of conditional expectation.
Therefore, we get
\begin{equation} \label{eq:amrtgbndedzero}
\ERJ[\again_{\ADVR}(t)|\statez,F]
\leq \ERJ[w_{max}(t) | \statez,F] + {\mathcal E}_0
\end{equation}
\noindent 
Using similar arguments for the expected gain of \ALGONAME,
\begin{equation} \label{eq:gainbndedzero}
 \ERJ[\gain_{\ALG}(t)|\statez,F]
\geq \ERJ[w_{max}(t)|\statez,F] \rho.
\end{equation}
 Summing the gains over time slots in the frame,
 % Taking conditional expectation with respect to the random decisions of our algorithm and the F-pattern $J$ on (\ref{eq:amortgain}) and using (\ref{eq:amrtgbndedzero}) for every slot 
 we have
 \begin{align*}
 &\EJ \Big[ \sum_{t=t_0}^{F} \gain_{\ADVR}(t) | \statez,F \Big] \leq \ERJ \Big[ \sum_{t=t_0}^{F} \again_{\ADVR}(t)| \statez ,F \Big]  \\
 &\leq \ERJ \Big[\sum_{t=t_0}^{t_0+F} w_{max}(t) | \statez,F\Big] + \mathcal E_0 F
 \end{align*}
 and taking the expectation with respect to frame size $F$,
 \begin{align}\label{eq1}
 &\EJ \Big[ \sumt \gain_{\ADVR}(t) | \statez \Big] \leq \ERJ[\sumt w_{max}(t) | \statez] + \bar{\mathcal E} 
 \end{align}
 where $\bar{\mathcal E}=a_{max}\E[F]\maxdead+2\E[F^2]$. Similarly,   
 \begin{gather} \label{eq2}
 \ERJ \Big[\sumt \gain_{\ALGONAME}(t)  | \statez \Big] \geq \rho \ERJ \Big[\sumt w_{max}(t)|\statez \Big]
 \end{gather}
 % we sum for every slot in (\ref{eq:gainbndedzero}) and we divide the two resulting inequalities obtaining
 % \begin{gather*} 
 % \frac {\ERJ[\sumt \gain_{ALG}(t)  | \statez]}{\EJ[\sumt \gain_{\ADVR}(t)]} \geq \\
 % \rho \frac { \ERJ[\sumt w_{max}(t)|\statez]}{\ERJ[\sumt w_{max}(t) | \statez] + \mathcal E F}
 % \end{gather*}
 % choosing $\epsilon$ small enough such that $\lambda_l-\epsilon>0$ yields
 % \[\Pr \lb  0< \frac {\sumt a_l(t) }   {F}\rb \geq 1-\delta\]
 % or equivalently
 % \[\Pr \lb  0< \sumt a_l(t)  \rb \geq 1-\delta \]
 % Therefore for every $\delta \in (0,1)$ there is $F_{0,l}$ such that there is at least one arrival for user $l$ with probability $\delta$ arbitrarily high.
 % Further since $a_l(t)$ is an integer quantity for each user and time-slot, we have
 % \[A_l := \left\{ 0 <  \sumt a_l(t)\right\}= \left\{ 1 \leq   \sumt a_l(t)\right\}\]
 Now consider link $l_1$ that has the maximum deficit at time $t_0$. At any time $t\in \mathcal{F}$,
 % \[\tau_a = \inf\{ t: a_{l_1}(t)\geq 1\} \in \frm\]
 % i.e. there is some random time $\tau_a$ when there is an arrival for user $l_1$. For that time the following holds:
 \[w_{l_1}(t_0)+a_{max}F \geq w_{l_1}(t) \geq w_{l_1}(t_0)-F. \]
 Recall that $w_{max}(t)$ denotes the maximum deficit among the nonempty links, and $a_{l_1(t)}> 0$ implies that the link $l_1$'s buffer is nonempty at time $t$. Therefore
  \begin{eqnarray}
  w_{max}(t) \geq w_{l_1}(t) \mathds{1}(a_{l_1}(t)>0) \geq w_{l_1}(t) \frac{a_{l_1}(t)}{\maxpacks}. 
  \end{eqnarray}
 % The first inequality holds by the fact that user $l_1$ has a packet therefore the weight of the maximum packet could correspond to user $l_1$ or some other user that has more weight.  The second inequality follows by the fact that for any time $t\in \frm$ and any user $l\in \users$ we have
 % \[ w_l(t_0) - F \leq w_l(t) \leq w_l(t_0) + \Delta \]
 Hence,
 \begin{align}
&  \ERJ[\sumt w_{max}(t)|\statez] \geq \ERJ \Big[ \sumt w_{l_1}(t)\frac{a_{l_1}(t)}{\maxpacks} \big| \statez\Big]
   \nonumber\\
  %&\geq \delta \ERJ \lb w_{max}(\tau_a) \big| \statez, A_{l_a}\rb \\
 & \geq \frac{1}{\maxpacks}\ERJ \Big[ (w_{l_1}(t_0)-F) \sumt a_{l_1}(t)\big| \statez \Big] \nonumber \\
 % &\geq w_{l_1}(t_0) \overline{a}_{l_1} \E[F] -a_{max} \E[F^2] \nonumber \\
&  \geq \frac{\|w(t_0)\|}{K}\E[F] \frac{\overline{a}_{l_1}}{{\maxpacks}}-\E[F^2]  \label{eq:arg-wmax-bound}
 \end{align}
 and therefore 
 $$\lim_{\|w(t_0)\| \to \infty} \ERJ \Big[\sumt w_{max}(t)|\statez \Big] = \infty.
 $$
 Using this and \dref{eq1} and \dref{eq2}, the result follows.
 From which it follows that
 \ben
  \frac {\ERJ[\sumt \gain_{\ALG}(t)  | \statez]}{\EJ[\sumt \gain_{\ADVR}(t) | \statez]} \geq \rho-\epsilon 
 %  &&\geq
 % \frac { \rho \ERJ[\sumt w_{max}(t)|\statez]}{\ERJ[\sumt w_{max}(t) | \statez] + \bar{\mathcal E} \E[F]} \\
 % &&\geq \frac{\rho (\|w(t_0)\|\overline{a}_{l_1}\E[F]/K-a_{max}\E[F^2])}{\|w(t_0)\|\overline{a}_{l_1}\E[F]/K-a_{max}\E[F^2])+ \bar{\mathcal E}} \to \rho
 \een
 as $\|w(t_0)\| \to \infty$.
\end{proof}
\begin{theorem} \label{mainthm}
For any policy $\ADV$, and \ALGONAME, given any $\epsilon>0$, there is $W'$ such that when $\|w(t_0)\| \geq W'$:
\[\E_{\statez} \Big[ \sumt \gain_{\ALGONAME}(t) \Big] \geq
(\rho-\epsilon) \E_{\statez} \Big[ \sumt \gain_{\ADV}(t) \Big] 
\]

\end{theorem}
\begin{proof}
 Using Lemma \ref{lemma:baselemma} for the optimal $\ADV$ over a frame $\frm$, and the fact that $\ADV$ is at least as effective as $\ADVR$
 \ben
   \E_{t_0}[\sumt \gain_{\ADV}(t)] \geq \E_{t_0}[\sumt \gain_{\ADVR}(t)] \geq \E_{t_0} [\sumt \gain_{\ADV}(t)] - a_{max} \E[F^2]
 \een
 Dividing by $\E_{t_0}[\sumt \gain_{\ADV}(t)]$
 and taking limits as $\|w(t_0)\| \rightarrow \infty$, the squeeze limits theorem yields: 
 \be \label{eq:limit2}
     \frac { \E_{t_0}[\sumt \gain_{{\ADVR}}(t)]}  
      {\E_{t_0}[\sumt \gain_{\ADV}(t)]} \rightarrow 1
 \ee
since,\NEW{as we showed in the proof of Lemma~\ref{lemmas0}}, $\E_{t_0}[\sumt \gain_{\ADV}(t)] \to \infty$, as $\|w(t_0)\| \to \infty$.
 Using~\dref{eq:limit2} and Lemma~\ref{lemmas0}, the result follows.
%The proof follows from Lemma~\ref{lemma:baselemma} and Lemma~\ref{lemmas0} and noting that $\E_{t_0}[\sumt \gain_{\ADV}(t)] \to \infty$, as $\|w(t_0)\| \to \infty$. We omit the details due to space constraint. 
\end{proof}

\subsection{Gain Analysis of {\MPALGONAME} in General Networks}\label{sec:gain_general}
First we show that binary search in Algorithm \ref{alg:general} suffices for computing $\nstar$ defined in (\ref{def:nstar}).
\begin{proposition}
The binary search in Algorithm \ref{alg:general} computes $\nstar$ as defined in (\ref{def:nstar}).
\end{proposition}
\begin{proof}
%We omit the proof due to space constraint.
 Assume that for some $n$, $p_n^n(t)\geq 0$. In this case we know that $\nstar \geq n$ since $n$ satisfies  (\ref{def:nstar}). Now assume that $p_{n}^{n}(t)<0$. Then we claim that we can conclude $\nstar<n$, or equivalently $p_{n'}^{n'}(t)<0$ for any $n'>n$.
 It suffices to prove that $p_{n}^{n}(t)<0$ implies $p_{n+1}^{n+1}(t)<0$, from which inductively the claim follows.
 To arrive at a contradiction, assume $p_{n}^{n}(t)<0$, $p_{n+1}^{n+1}(t)\geq 0$, or equivalently $(a)$: $\subhrm{n} > \misw{n}$ and $(b)$: $\subhrm{n+1} \leq \misw{n+1}$.
 Then
 \be
 &\frac{1}{\misw{n+1}} - \frac{1}{n \misw{n+1}} \overset{(b')}{\leq} \frac{1}{\subhrm{n+1}} -  \frac{1}{n\misw{n+1}} = \nonumber\\
 &\frac{\sumd{n+1} \misw{i}^{-1}}{n}
  - \frac{1}{n \misw{n+1}}
 = \frac{\sumd{n} \misw{i}^{-1}}{n} = \nonumber \\
 &\frac{n-1}{n}\frac{\sumd{n} \misw{i}^{-1}}{n-1} 
 \overset{(a')}{<} \frac{n-1}{n} \frac {1}{\misw{n}}, \nonumber
 \ee
 where in $(a')$ we used $(a)$ and in $(b')$ we used $(b)$. This shows $\frac {1}{\misw{n+1}} < \frac {1}{\misw{n}}$ or
 $\misw{n+1} >\misw{n}$, which is a contradiction with the ordering of $M_i$. 
Hence $p_{n}^{n}(t)<0$ implies $p_{n+1}^{n+1}(t)<0$. 
\end{proof}
We next state Lemmas~\ref{lemma:subhrm-monotonic} and \ref{lemma:mp-expected-gaininequality} regarding  the properties of the probabilities used by {\MPALGONAME}, which are used in the gain analysis. % derivation of Lemma \ref{lemma:mp-expected-gaininequality}.
Their proofs follow directly from the probabilities used by {\MPALGONAME}. 
\begin{lemma} \label{lemma:subhrm-monotonic}
$\subhrm{n}$ (defined in \dref{def:subhrm}) is strictly decreasing as a function of $n$, for $\nstar \leq n\leq \numbMIS$. 
\end{lemma}
 \begin{proof}
 %We omit the proof due to page constraint.
 Take any $n$, $\nstar < n\leq \numbMIS$. By the definition of $\nstar$ it must be the case that $p_{n}^n(t)<0$, which implies $\misw{n} < \subhrm{n}$.
 From this, and by using \dref{def:subhrm},
 \be
 &\misw{n}^{-1} (n-1) > \sumd{n} \misw{i}^{-1} \label{eq:helpbnd10}.
 \ee
 We then have
 \begin{flalign*}
 & \sumd{n} \misw{i}^{-1} = \sumd{n-1} \misw{i}^{-1}+\misw{n}^{-1}\\
 &=  \sumd{n-1} \misw{i}^{-1}+ \frac{n-1}{n-2} \misw{n}^{-1}-
  \frac{\misw{n}^{-1}}{n-2}
 \\ 
 &\overset{(a)}{>}\sumd{n-1} \misw{i}^{-1}+ \frac{1}{n-2} \sumd{n} \misw{i}^{-1} - \frac{\misw{n}^{-1}}{n-2} 
 \nonumber \\
 &= \sumd{n-1} \misw{i}^{-1}+  \frac{1}{n-2} \sumd{n-1}  \misw{i}^{-1} \\
 &= \frac {n-1}{n-2}\sumd{n-1}  \misw{i}^{-1},
 \end{flalign*}
 where in $(a)$ we used (\ref{eq:helpbnd10}). Dividing both sides by $n-1$, we get $\subhrm{n}^{-1} > \subhrm{n-1}^{-1}$. 
\end{proof} 
%\noindent
%We prove the following lemma that will be used for the derivation of subsequent results.
%The next property is stated below. 
\begin{lemma} \label{lemma:mp-expected-gaininequality}
If $i\not\in [\nstar]$ and $j\in [\nstar]$, for the choice of probabilities $p_{k}^{\nstar}(t)$ in (\ref{eq:mwisprob}) selected by $\MPALGONAME$, we have
\be
&\misw{i}+ \sumdk{\nstar} p^{\nstar}_k(t) \misw{k} < \nonumber \\ 
& \misw{j}+\sum_{k \in [\nstar]\setminus\{j\}} p^{\nstar}_k(t) \misw{k} \nonumber
\ee
\end{lemma}
 \begin{proof}
 %Proof is omitted due to space constraint. 
 %Since $i\not\in [\nstar]$ then for some $k>\nstar$ we must have $p_i^{k}\leq 0$. But since $\subhrm{n}$ is decreasing for $n\geq \nstar$ we must also have $p_{i}^{\nstar+1}\leq 0$.
 %Partition all $j\in \numMIS \setminus [L(t)]$ into sets $\{P_k\}$ corresponding to the iteration  of the algorithm in which they were eliminated. For example $P_1$ will contain all the \MIS\ eliminated in iteration $1$. Say that the algorithm needed $k$ iterations to terminate.  consider the \MIS\ in $P_i$, and 
 %Since $j\not\in [L(t)]$ it must have been eliminated in some iteration $k$ of the algorithm. Therefore at stage $k$ the set $[L_k]$ we were examining contained $j$ and it was eliminated because it had nonpositive probability $p_j^{L_k}(t)\leq 0$. 
 Equivalently after simplifying the inequality, we need to prove:
 \[\misw{i} < \misw{j} (1-p_j^{\nstar}(t)) = \subhrm{\nstar}.\]
 Since $i\not\in[\nstar]$, we have $\misw{i}< \subhrm{i} $,
 and from the monotonicity of $\subhrm{n}$ for $n\geq \nstar$ (Lemma~\ref{lemma:subhrm-monotonic}), since $i>\nstar$, we have $\subhrm{i} < \subhrm{\nstar}$. Therefore, $\misw{i} < \subhrm{\nstar}$.
 \end{proof}
\begin{lemma}\label{lemma:boundsadvalg-mp}
For each time $t\in \frm$, the gain obtained by \MPALGONAME, and the amortized gain obtained by the Max-Gain policy $\ADV$, starting from some state $\statet$, satisfy:
\be
\ER[\again_{\ADV}(t)|\statet] \leq  \sumd{\nstar} \misw{i} - (\nstar-1) \subhrm{\nstar} + \extracompmul\\
 \ER[\gain_{\MPALGONAME}(t)|\statet] = \sumd{\nstar} \misw{i} - \nstar \subhrm{\nstar}
\ee
where $\extracompmul=KF$and $\ER$ is with respect to decisions of \MPALGONAME.
\end{lemma}
\begin{proof}
%Assume that the complete N-partite graph is partitioned to the noninterfering components $V_i$ such that $K=\bigcup_{i=1}^{i=N} V_i$ and assume that $s_1\geq s_2 \geq ...\geq s_N$.
%Consider the largest $L$ (where $L\leq K$) with the property 
%\[p_i=1-\frac{c^L}{s_i} \geq 0, \forall i\in [L]\]
%Let $K'=\bigcup_{i=1}^{i=L} V_i$
Using the probabilities computed by \MPALGONAME, the expected gain of {\MPALGONAME} at time $t$ is 
\[
\E[\gain_{\MPALGONAME}(t)] = \sumd{\nstar} p_i^{\nstar}(t) \misw{i} = 
\sumd {\nstar} \misw{i} - \nstar \subhrm{\nstar}
\]
Next for the amortized gain of the Max-Gain Policy $\ADV$, we will apply the same technique as in the collocated networks case, where we modify the buffers and give $\ADV$  additional reward.
Suppose $\mu$ transmits $M_i$, and {\MPALGONAME} transmits some $M_j$. 
    We make the buffers the same by allowing $\ADV$ to additionally transmit all the packets that are transmitted by {\MPALGONAME} but not by $\ADV$ (i.e., in links $M_j\setminus M_i$). Since this will result in a decrease of the deficit by one for each link in $M_j\setminus M_i$ for $\ADV$ in the remaining slots, we give $\ADV$ an additional reward $\extracompmul = KF$ which is an upper bound on the number of packets transmitted by $\ADV$ from links $M_j\setminus M_i$ in the remaining slots. To compute the expected gain, we differentiate between two cases: 
    
%\begin{enumerate}[leftmargin=*]
    \textbf{Case 1.} $i\in [\nstar]$. In this case, we can write
    \begin{flalign}
    &\E[{\again_{\ADVR}}^{M_i}(t)] {=} \misw{i}      + \sum_{j \in [\nstar]\setminus\{i\}} p_j^{\nstar}(t) \left( W_{M_j\setminus M_i}(t)+\extracompmul \right)  \nonumber 
    \\
    &{\leq} \misw{i} + \sum_{j \in [\nstar]\setminus\{i\}} p_j^{\nstar}(t) \left(\misw{j}+\extracompmul \right)  \label{eq:mp-usfl-ineq1}\\
    % &\leq \misw{i} + \sum_{j \in [\nstar]\setminus\{i\}} p_j^{\nstar}(t) \misw{j} +\extracompmul  \label{eq:mp-usfl-ineq1} \\
    &=
    \misw{i}(1-p_i^{\nstar}(t)) + \sum_{j \in [\nstar]} p_j^{\nstar}(t) \misw{j} + \sum_{j \in [\nstar]\setminus\{i\}} p_j^{\nstar}(t)\extracompmul \nonumber
    \\
    &= \subhrm{\nstar} + \sumd{\nstar} \misw{i} - \nstar \subhrm{\nstar} + \sum_{j \in [\nstar]\setminus\{i\}} p_j^{\nstar}(t)\extracompmul \label{eq:intermediate-step-used}
    \\
    &\leq \subhrm{\nstar} + \sumd{\nstar} \misw{i} - \nstar \subhrm{\nstar} + \extracompmul.
    \label{eq:mp-usfl-ineq2}
    \end{flalign}
    
    \textbf{Case 2.} $i\not \in [\nstar]$. In this case, we have
    \be
    \E[{\again_{\ADVR}}^{M_i}(t)] \leq &\misw{i} + \sumdk{\nstar} p_k^{\nstar}(t) (\misw{k}+\extracompmul) \nonumber \\
    \stackrel{a}{<}&\misw{j}+\sum_{k \in [\nstar]\setminus\{j\}} p_k^{\nstar}(t) \misw{k} + \extracompmul \nonumber \\
    = &
     \subhrm{\nstar} + \sumd{\nstar} \misw{i} - \nstar \subhrm{\nstar} + \extracompmul, \nonumber
    \ee
    where in $(a)$ we applied Lemma \ref{lemma:mp-expected-gaininequality} for $i,j$.
    Note that in both cases, the upper bound is the same and does not depend on the particular choice of $M_i$.
\end{proof}
\begin{lemma}\label{lemma:bound-general-case}
For $\subhrm{\nstar}$ in \dref{def:subhrm}, We have
\ben
&\frac{\sumd{\nstar} \misw{i} - \nstar \subhrm{\nstar}}
{ \sumd{\nstar} \misw{i} -(\nstar-1)\subhrm{\nstar} } \geq 
\frac {\numbMIS}{2\numbMIS-1}.
\een
\end{lemma}
\begin{proof}
Suffices to show that
 \be \label{inequality-multi}
 \frac{\sumd{\nstar} \misw{i} - \nstar \subhrm{\nstar}}
 {\sumd{\nstar} \misw{i} -(\nstar-1)\subhrm{\nstar} }
 \geq 
 \frac {\nstar}{2\nstar-1}
 \ee
Since by definition $\numbMIS \geq \nstar$. For the non-trivial case, we have $\nstar-1> 0$, and therefore inequality (\ref{inequality-multi}) can equivalently be written as
 $
 (\nstar-1)\sumd{\nstar} \misw{i} \geq 
 \nstar^2(t) \subhrm{\nstar}.  
 $
This inequality holds since it follows by applying the inequality between arithmetic and harmonic means:
\ben
\frac{1}{\nstar} {\sumd{ \nstar} \misw{i} }  \geq 
\frac{\nstar}{\sumd{\nstar} \misw{i}^{-1}},  
\een
and the fact that $\nstar-1 \geq 1$.
\end{proof}

\begin{theorem} \label{mainthm-mp}
Under \MPALGONAME, given any $\epsilon>0$ there is W' such that for all $\|w_0\|=\suml w_L(t_0) \geq W'$,
\[
 \ERJ_{t_0} \Big[ \sumt \mathcal G_{\MPALGONAME}(t) \Big] 
\geq (\rho-\epsilon)\EJ_{t_0} \Big[ \sumt \mathcal G_{\ADV}(t)\Big] ,  
\]
where $\mu$ is any non-causal policy, and $\rho= \frac {|\allMIS|}{2|\allMIS|-1}$.
\end{theorem}
\begin{proof}
By using Lemma~\ref{lemma:boundsadvalg-mp}, summing and taking expectation similar to the proof of Lemma~\ref{lemmas0}, it follows that
%and we get\javad{expectations are missing from some of the expressions below} \TODO{\checkmark}
\ben
&\E \Big[ \sumt \gain_{\MPALGONAME}(t)|\statez \Big] = \E \Big[ \sumt x(t) | \statez \Big] \label{eq:general-gain1}\\
&\E \Big[ \sumt \gain_{\ADVR}(t)|\statez \Big] \leq \bar{\mathcal E}_m + \E \Big[\sumt  y(t)| \statez \Big]  \label{eq:general-gain2}
\een
where $\bar{\mathcal E}_m=K\E[F^2]$, and $x(t)=y(t)-\subhrm{\nstar}$, where
\ben
%&x(t):=\sum_{i \in \nstar(t)} \misw{i}(t) - \subhrm{\nstar} \nstar(t)\\
&y(t) := \summis \misw{i} - (\nstar(t)-1)  \subhrm{\nstar}.
\een
%It is easy to verify that $y(t) \geq \misw{1} \geq w_{max}(t)$. Letting $\|w_0\|\rightarrow \infty$, and using Lemma~\ref{lemma:bound-general-case} yields the result. We omit the details due to page constraint.
 Now notice that 
 \NEW{
 \be
 y(t) &=& \subhrm{\nstar} + \sumd{\nstar} \misw{i} - \nstar \subhrm{\nstar} 
 \nonumber \\
& =& \misw{1}(1-p_1^{\nstar}(t)) + \sum_{i \in [\nstar]} p_i^{\nstar}(t) \misw{i}  \nonumber\\
&=& \misw{1} + \sum_{i \in [\nstar]\setminus\{1\}} p_i^{\nstar}(t) \misw{i}  \nonumber \\
 &\geq&
 \misw{1} \geq w_{max}(t) \label{eq:arg-general-wmax}
 \ee
 }
%  where $(a)$ can be seen  \javad{???} \TODO{edited. Changed to equality} by the expressions in (\ref{eq:mp-usfl-ineq1}) and (\ref{eq:intermediate-step-used}) for $M_1$ after the cancellation of the common term. 
Now notice
 \be
 &&\lim_{\|w_0\|\rightarrow \infty} \frac{ 
 \E [ \sumt  x(t) | \statez ]}
 {\E [ \sumt y(t) | \statez ] +\bar{\mathcal E}_m } \nonumber \\
 && \overset{(a)}=
 \lim_{\|w_0\|\rightarrow \infty}
  \frac{ 
 \E [ \sumt  x(t) | \statez ]}
 {\E [ \sumt y(t) | \statez ]} 
 \overset{(b)}{\geq} 
 \frac {|\allMIS|}{2|\allMIS|-1} \nonumber,
 \ee
 where in $(a)$ we used the fact that $\bar{\mathcal E}_m<\infty$, and that the remaining expression in the denominator goes to infinity using the inequality derived in (\ref{eq:arg-general-wmax}) alongside the argument in (\ref{eq:arg-wmax-bound}).
 In $(b)$ we used Lemma~\ref{lemma:bound-general-case}.
\end{proof}

\section{Simulation Results}\label{sec:sim}
% \NEW{Initially we remark that this problem is challenging because we are trying to balance between collecting as much deficit as possible, while we improve our future state (so that we have packets to transmit in future slots). If we always have packets for every link or the deadlines of all packets were $\infty$, then the problem would have been reduced to the regular non real-time scheduling problem and LDF would have been optimal. For example}
\NEW{
If the packet arrival rate becomes very large, any policy inevitably will be restricted to a small delivery ratio $\mathbf p$. But then due to high availability of packets in the buffers, the policy can always schedule packets, thus leading to a small deficit queue under such small $\mathbf p$, even for simple and naive policies. Hence, the problem is interesting and challenging when the packet arrival rate is not too high so that the optimal policy can fundamentally achieve a high $\mathbf p$. Similarly, if the packet deadlines become very large, the problem is reduced to the regular non real-time scheduling and deadline-oblivious algorithms like LDF should perform reasonably well. Hence, we focus on the interesting scenario when packet arrival rates or deadlines are not excessively large.}

In our simulations, 
% arrival with a large availability of packets. In this context the deadlines are less relevant since there are many packets that can be transmitted, making LDF for example in the case of collocated networks near optimal, since as we mentioned earlier, in the case where there are no deadlines LDF is optimal in collocated networks. On the other hand, given the current model, 
% the problem is interesting and challenging for cases where the optimal solution can achieve high $\mathbf p$.
we consider two cases for the deficit admission (see the model section): one is based on coin tossing where each arrival on a link $l$ is counted as deficit with probability $p_l$, and the other is deterministic, where each arrival increases the deficit by exactly $p_l$.

We compare the performance of our randomized algorithms, {\ALGONAME} and {\MPALGONAME} with LDF. Recall that LDF chooses the longest-deficit link, then removes the interfering links with this link, and repeat the procedure. We further consider two versions of LDF: One is LDF that does a random tie breaking when presented with a deficit tie (LDF-RD), and the other version tries to schedule the non-dominated link and its earliest-deadline packet (LDF-ED) in such tie situations. In the plots, we compare the average deficit (over all links) as we vary the value of the delivery ratio. 
 \begin{figure}
 \centering
 \begin{subfigure}[b]{.20\textwidth}
   \centering
   \includegraphics[width=0.95\linewidth]{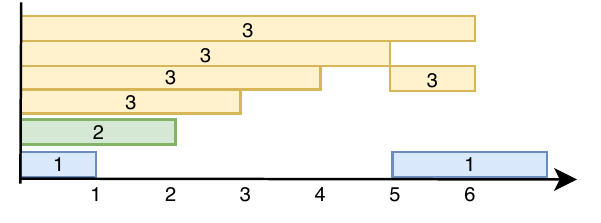} \caption{Traffic pattern F}
     \label{fig:depict-trafficF}
 \end{subfigure}%
 \begin{subfigure}[b]{.30\textwidth}
   \centering
   \includegraphics[width=0.95\linewidth]{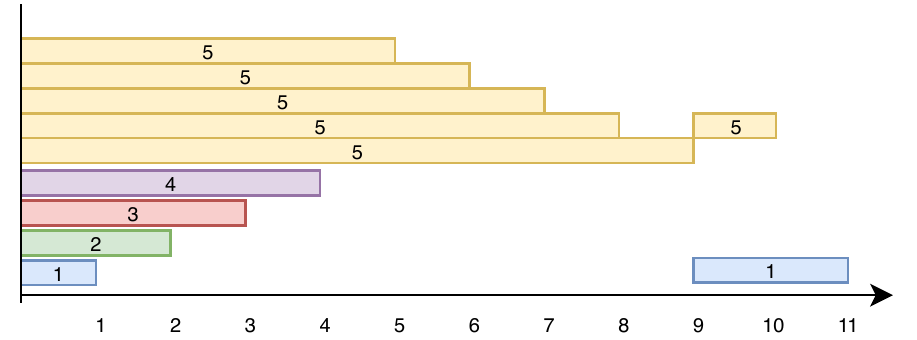}
   \caption{Traffic pattern A}
     \label{fig:depict-trafficA}
 \end{subfigure}
 \caption{Traffic patterns used in simulations}
 \label{fig:depictions-of-traffics}
 \end{figure}

\begin{figure}[t]
\centering
\begin{subfigure}{.25\textwidth}
  \centering
  \includegraphics[width=0.95\linewidth]{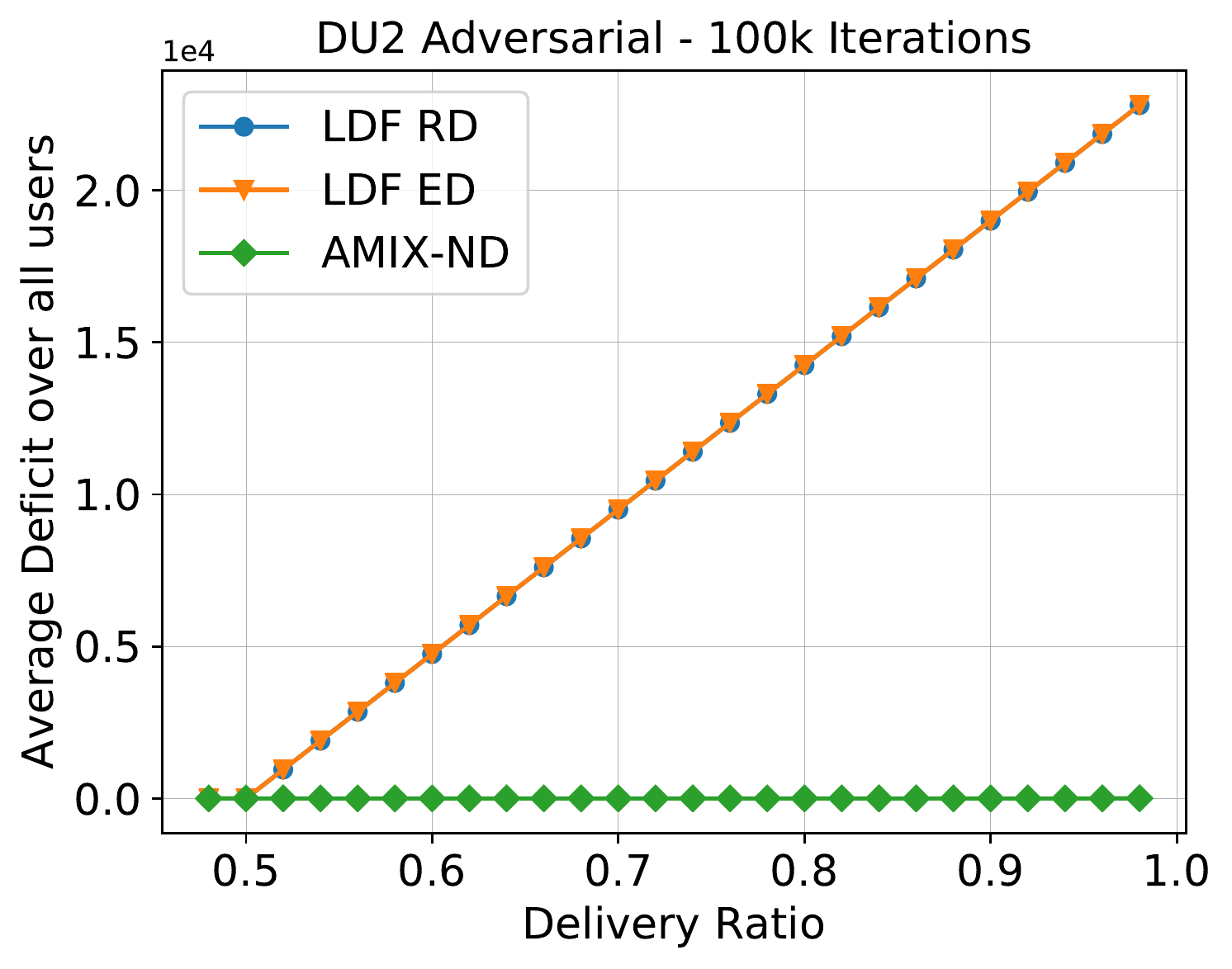} \caption{Deterministic deficit admission}
  \label{fig:col-biggap1}
\end{subfigure}%
\begin{subfigure}{.25\textwidth}
  \centering
  \includegraphics[width=0.95\linewidth]{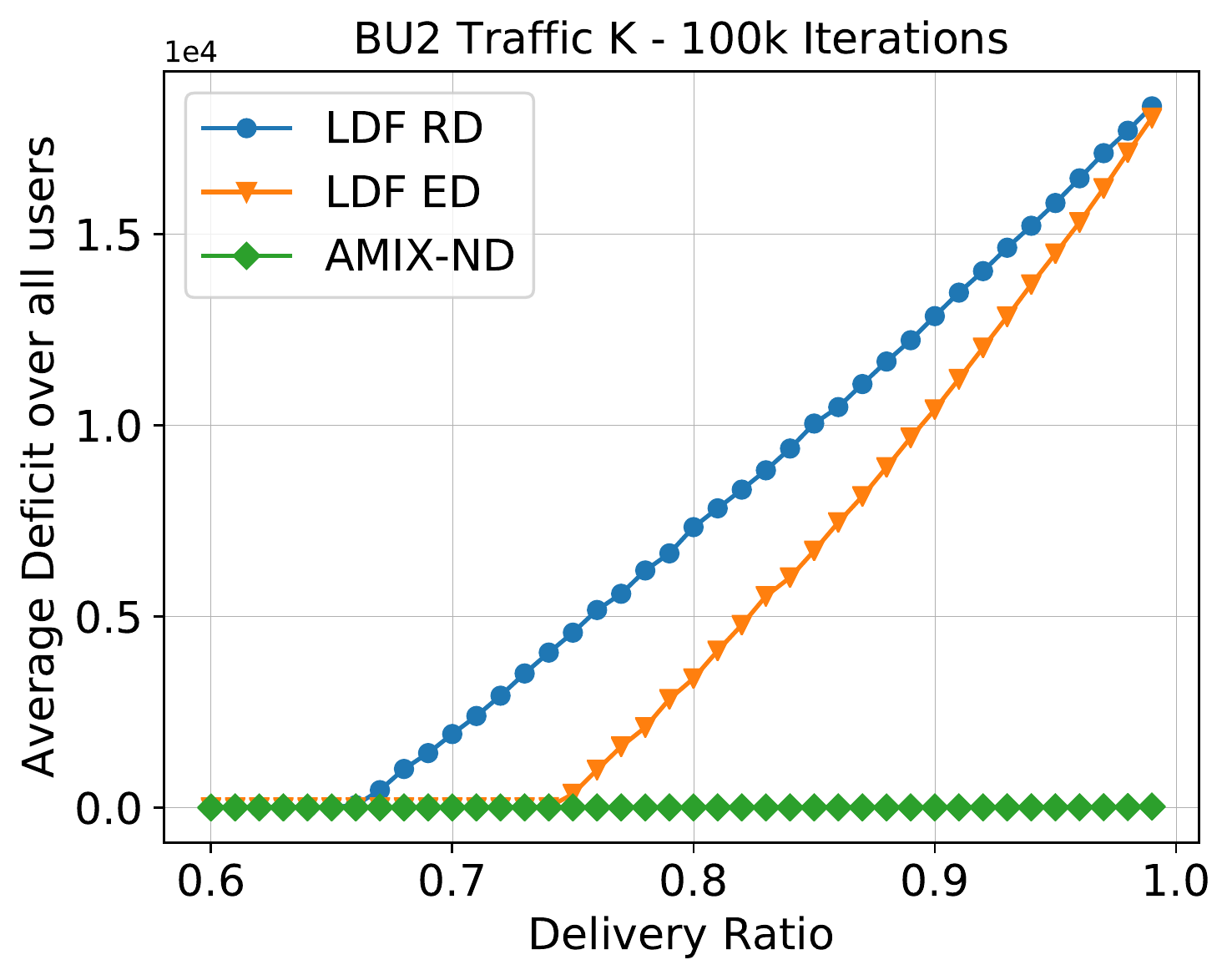}
  \caption{Coin-tossing deficit admission}
  \label{fig:col-biggap2}
\end{subfigure}
\caption{Comparison between {\ALGONAME} and LDF policies in a two-link network.}
\label{fig:generalgraph-comparison}
%\vspace{- 0.2 in}
\end{figure}

 \begin{figure}[t]
 \centering
 \begin{subfigure}{.25\textwidth}
   \centering
   \includegraphics[width=0.95\linewidth]{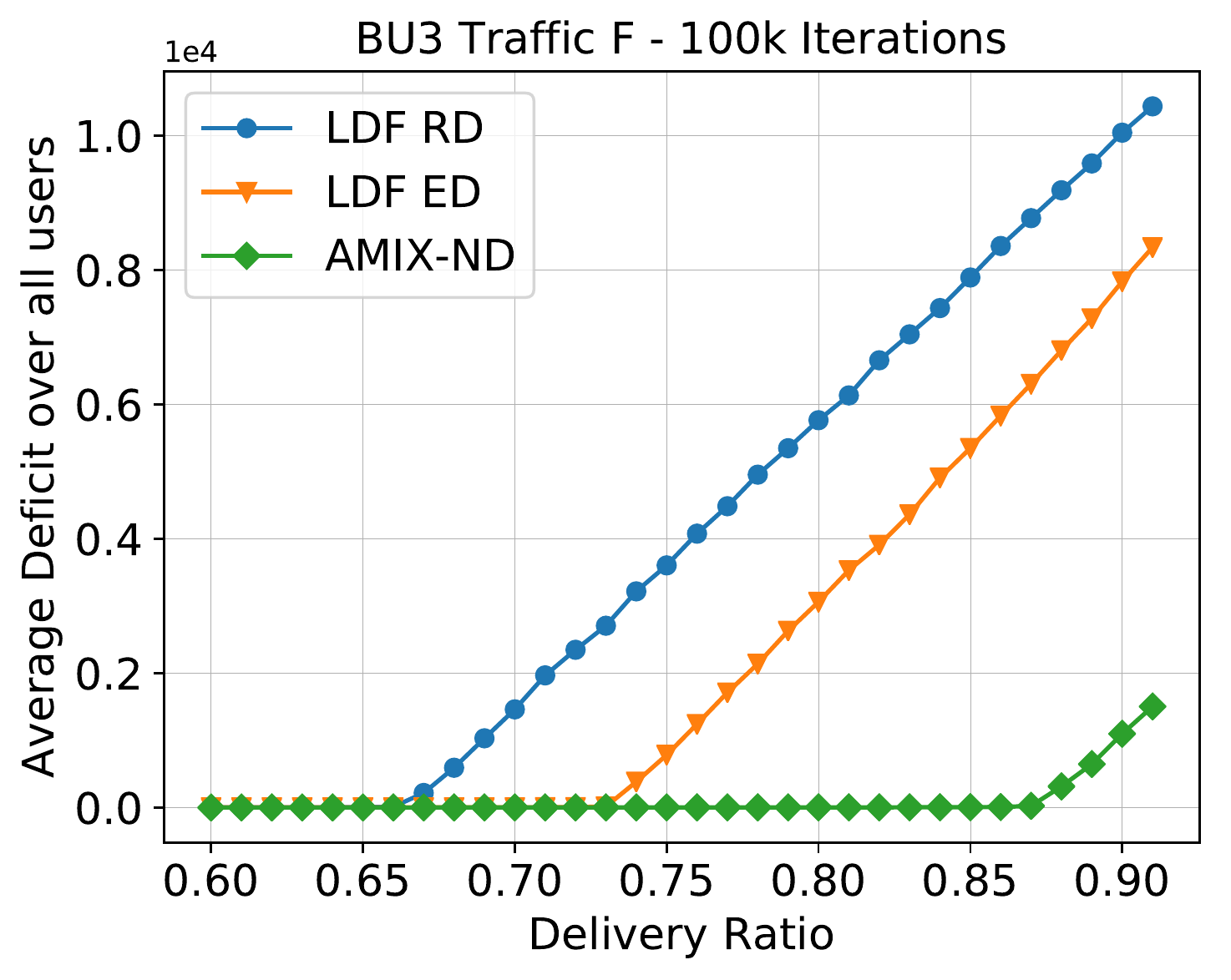} \caption{A collocated network with 3 users.}
   \label{fig:col-trafficF}
 \end{subfigure}%
 \begin{subfigure}{.25\textwidth}
   \centering
   \includegraphics[width=0.95\linewidth]{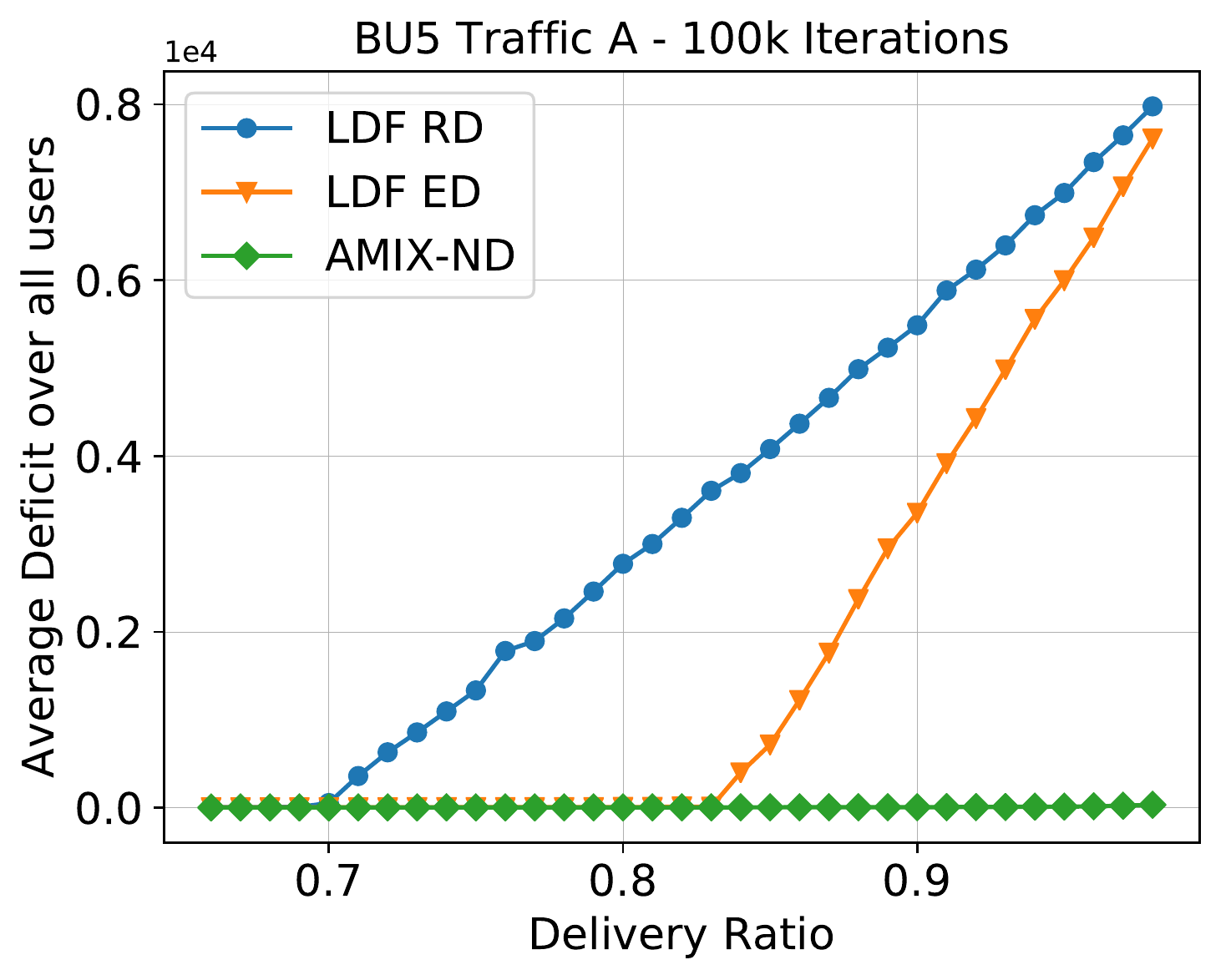}
   \caption{A collocated network with 5 users.}
   \label{fig:col-trafficA}
 \end{subfigure}
 \caption{Comparison between {\ALGONAME} and LDF policies in collocated networks under coin-tossing deficit admission.}
 \label{fig:col-more}
 \end{figure}

\begin{figure}[t]
\centering
\begin{subfigure}{.25\textwidth}
  \centering
  \includegraphics[width=0.95\linewidth]{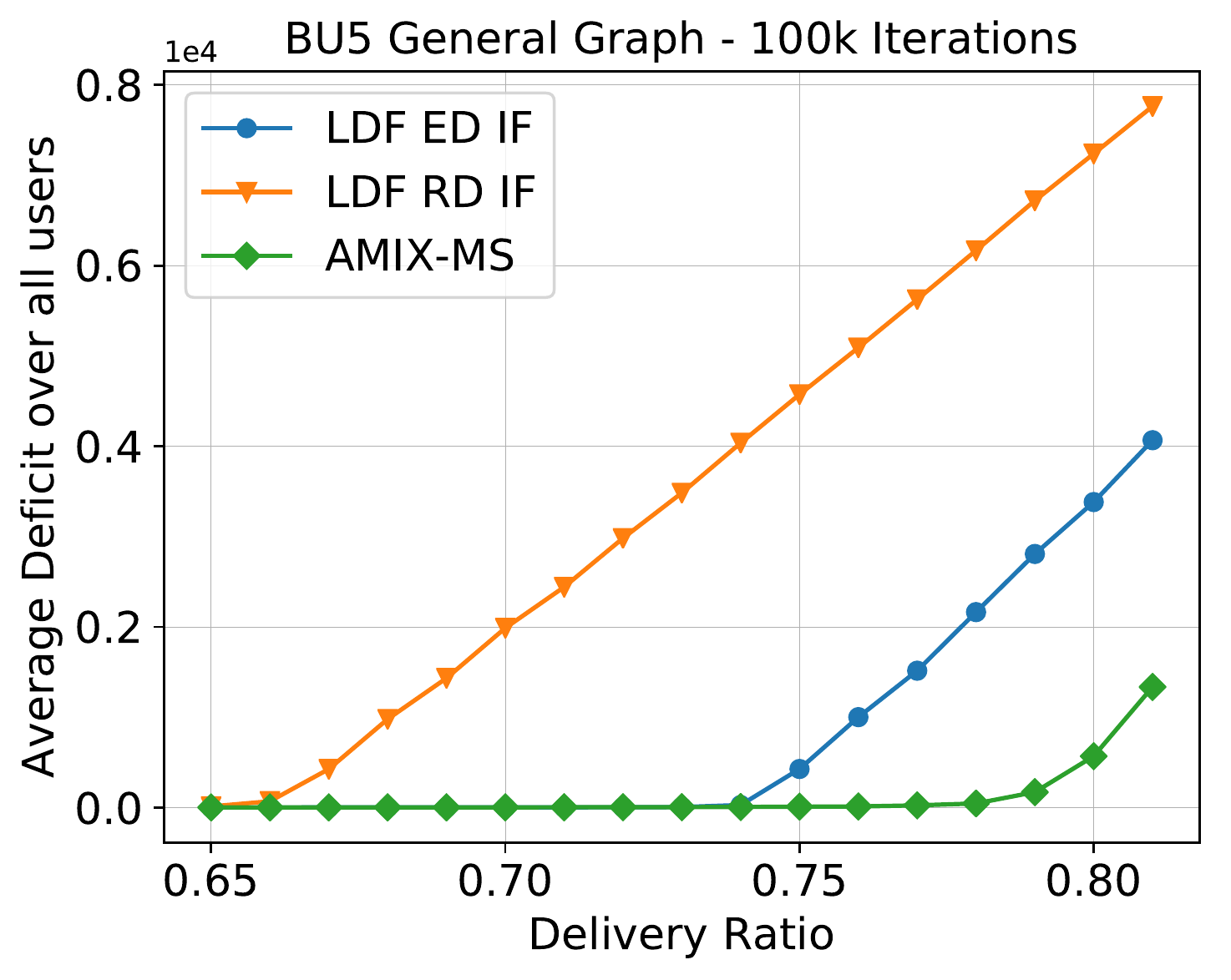}
  \caption{Coin-tossing deficit admission}
  \label{fig:simple-graph-sub1}
\end{subfigure}%
\begin{subfigure}{.25\textwidth}
  \centering
  \includegraphics[width=0.95\linewidth]{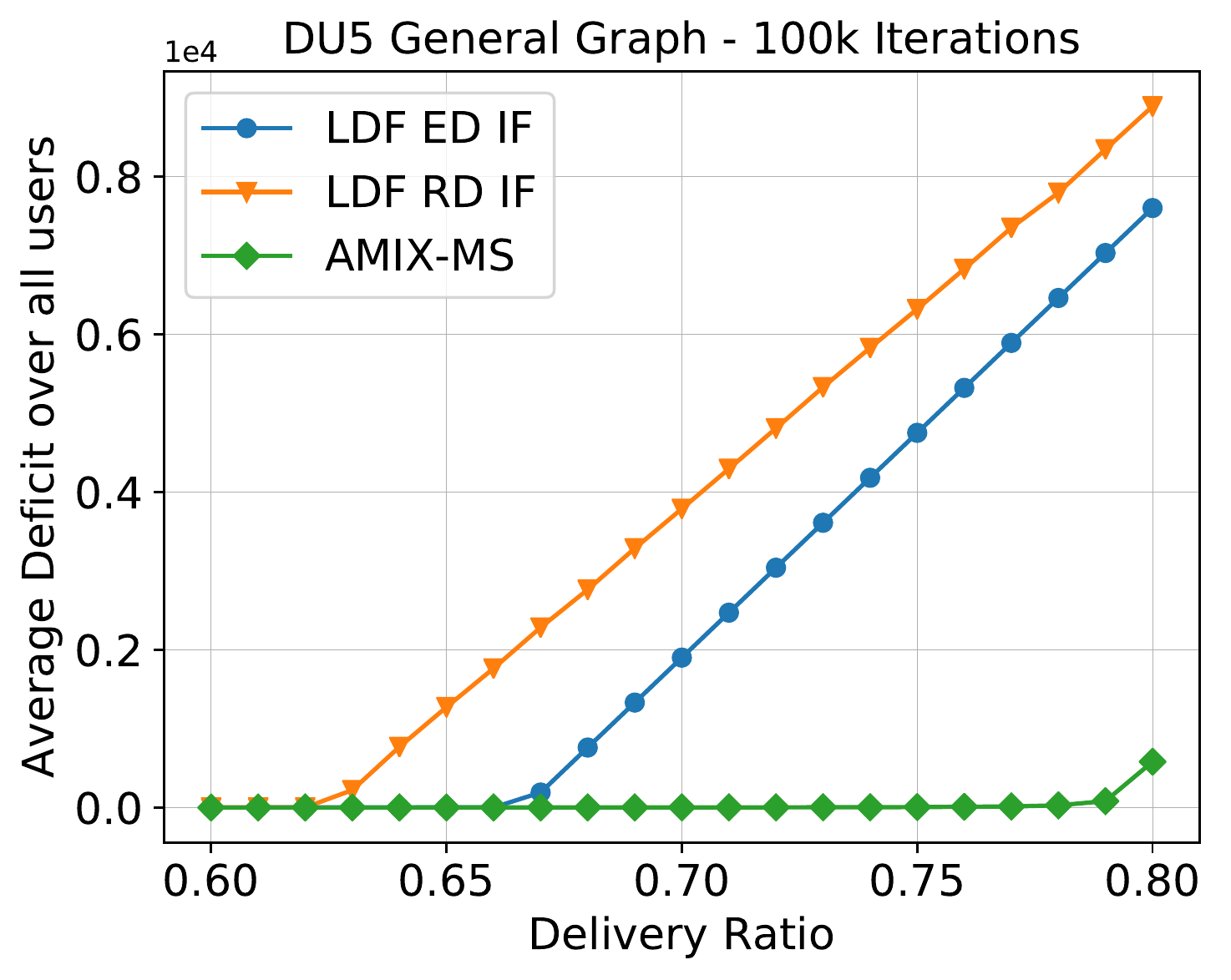}
  \caption{Deterministic deficit admission}
  \label{fig:simple-graph-sub2}
\end{subfigure}
\caption{Comparison between {\MPALGONAME} and LDF policies in a lightly connected interference graph with 5 links.}
\label{fig:simple-graph}
%\vspace{- 0.1 in}
\end{figure}
\textit{Collocated Networks.} We first consider two interfering links with deterministic deficit admission. The traffic is periodic and consists of alternating Pattern A and Pattern B of Figure~\ref{fig:tightexample}, with the delivery ratios satisfying $p_2=p_1+0.001$. Figure~\ref{fig:col-biggap1} shows the result. As we can see, \ALGONAME\ is able to achieve roughly $p_1 = 0.996$, whereas both versions of LDF become unstable for $p_1=0.5+\epsilon$.
In Figure~\ref{fig:col-biggap2}, again for two users, we used a traffic that consists of Pattern C followed by Pattern B, repeatedly. This time we keep $p_1=p_2$. \ALGONAME\ achieved near $p_1=1.0$, whereas the better version of LDF achieved roughly $0.75$, resulting in a gap of around $0.25$. 

\NEW{Figure~\ref{fig:col-trafficF} and Figure~\ref{fig:col-trafficA} show the results for collocated networks with various number of users, when traffic F and traffic A from Figure~\ref{fig:depictions-of-traffics} are used, respectively. In Traffic F, when $p_1=p_2=p_3=p$, the optimal policy can support at most $p=7/8 = 0.875$. In this case \ALGONAME\ achieves at least $p=0.87$, whereas LDF-ED achieves roughly $p=0.73$. Traffic A is similar in nature, but with more users and \ALGONAME\ is able to transmit all the packets; the result is shown in Figure~\ref{fig:col-trafficA}. 
}

\textit{General Networks.} We first consider the interference graph $\mathcal{G}_1$ in Figure~\ref{fig:simpleifgraph} involving 5 links, and interference edges $E_l=\{(l_1,l_2),(l_2,l_3),(l_2,l_4),(l_4,l_5)\}$.  For links $l_2$ and $l_5$, we have a periodic traffic with period $t=5$, where in slot 1 there are 2 packets arriving with deadline 2 and 3 and in slot 4 a packet arrives with deadline 1, and for links $l_1,l_3,l_5$, we have 1 packet arriving with deadline 1 at slot 1, and 1 packet arriving with deadline 2 at slot 4. The result for this graph is shown in Figure~\ref{fig:simple-graph}.

\NEW{
Next, we consider a complete bipartite graph $\mathcal{G}_2$ with two components, $V_1=\{l_1^\prime,l_2^\prime,l_3^\prime,l_4^\prime\}$ and $V_2=\{l_5^\prime,l_6^\prime,l_7^\prime,l_8^\prime\}$. The traffic used for links $l_1^\prime,l_2^\prime$ is the same as that of link $l_1$ in Graph $\mathcal{G}_1$ above. For links $l_3^\prime,l_4^\prime$ we used 
i.i.d. Bernulli with 1 arrival having deadline 1 with probability 0.25. For links $l_5^\prime,l_6^\prime$ we used the traffic used for link $l_2$ in Graph $\mathcal{G}_1$. For links $l_7^\prime,l_8^\prime$ we used i.i.d. traffic with $7$ arrivals with probability $0.05$, and $0$ arrivals otherwise, and deadline $10$. The results are depicted in Figures~\ref{fig:bipartite-6u-sub1} and \ref{fig:bipartite-6u-sub2}. 
}

\begin{figure}[t]
\centering
        \includegraphics[width=0.15\textwidth]{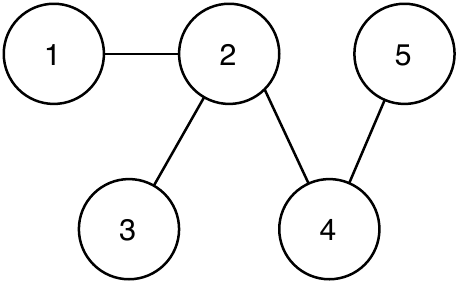}
\caption{Interference graph $\mathcal{G}_1$ used for simulations in Figure~\ref{fig:simple-graph}.}
\label{fig:simpleifgraph}
\end{figure}

 \begin{figure}[t]
 \centering
 \begin{subfigure}{.25\textwidth}
   \centering
   \includegraphics[width=0.95\linewidth]{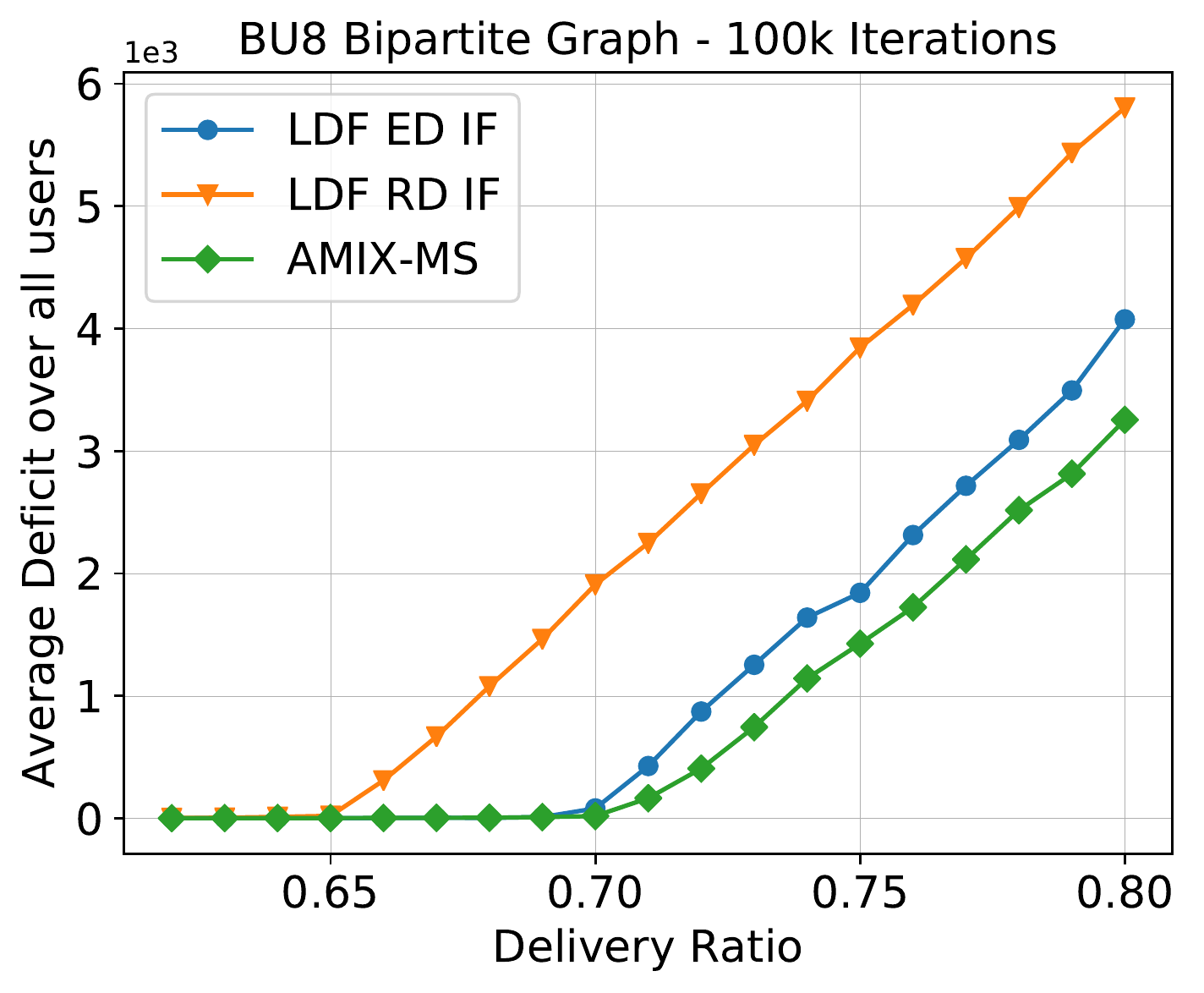}
   \caption{Coin-tossing admission}
   \label{fig:bipartite-6u-sub1}
 \end{subfigure}%
 \begin{subfigure}{.25\textwidth}
   \centering
   \includegraphics[width=0.95\linewidth]{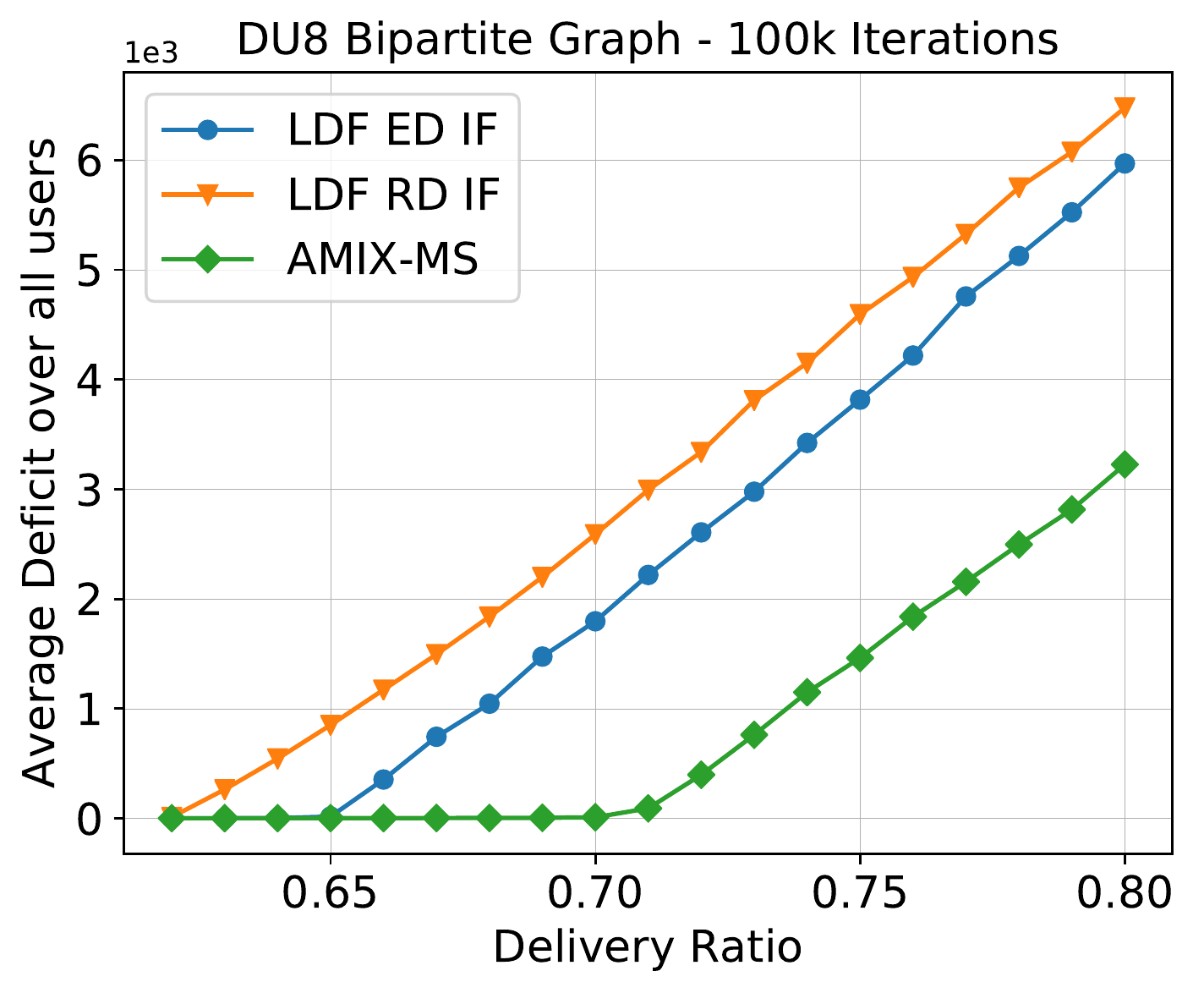}
   \caption{Deterministic admission}
   \label{fig:bipartite-6u-sub2}
 \end{subfigure}
 \caption{Comparison between policies on a complete bipartite graph with 8 links, and i.i.d. and Markovian arrivals.}
 \label{fig:bipartite-6u}
 \end{figure}

%\begin{figure}
%\centering
%        \includegraphics[width=0.30\textwidth]{simulations-general/AverageDeficitOver70kslots-BU5SimpleGraph.pdf}
%\caption{Here I used the previous graph.}
%\label{fig:simpleifgraph-b}
%\end{figure}
%\begin{figure}
%\centering
%        \includegraphics[width=0.30\textwidth]{simulations-general/AverageDeficitOver50kslots-DU5SimpleGraph.pdf}
%\caption{Deterministic Deficit Admission}
%\label{fig:simpleifgraph-d}
%\end{figure}

% \begin{figure}
% \centering
%         \includegraphics[width=0.30\textwidth]{simulations-general/Figure75.png}
% \caption{(THIS WILL BE REMOVED) High Variance Depiction - Case 1}
% \end{figure}
% \begin{figure}
% \centering
%         \includegraphics[width=0.30\textwidth]{simulations-general/Figure73.png}
% \caption{(THIS WILL BE REMOVED) High Variance Depiction - Case 2}
% \end{figure}
%
%\begin{figure}
%\centering
%        \includegraphics[width=0.30\textwidth]{simulations-general/AverageDeficitOver50kslots-BU6Bipartite.pdf}
%\caption{Bernulli Deficit Admission. Bipartite Graph with 6 Users. \TODO{I will add more details}}
%\label{fig:gen.b.u6}
%\end{figure}
%\begin{figure}
%\centering
%        \includegraphics[width=0.30\textwidth]{simulations-general/AverageDeficitOver50kslots-DU6Bipartite.pdf}
%\caption{Deterministic Deficit Admission. Bipartite Graph with 6 Users}
%\label{fig:gen.d.u6}
%\end{figure}

%\subsection{Discussion}

As we see, simulation results indicate that there are many scenarios that result in significant gap between our algorithms and LDF variants. This gap is especially pronounced when deterministic deficit admission is used, which is preferable as it provides a short-term guarantee on the deficit of a user. 

% The difference between the two could be attributed to the implicit randomization introduced in the algorithm through the random deficit admission or simply to the choice of traffics that were simulated. On the other hand there are cases where LDF ED performs in many traffic-scenarios similarly with our algorithms, therefore we do think there are cases where LDF ED could be suitable. 

\section{Conclusion}
In this paper, we studied real-time traffic scheduling in wireless networks under an interference-graph model. Our results indicated the power of randomization over the prior deterministic greedy algorithms for scheduling real-time packets. In particular, our proposed randomized algorithms significantly outperform the well-known LDF policy in terms of efficiency ratio. As a future work, we will investigate efficient and distributed implementation of {\MPALGONAME} for general graphs.

%\clearpage
\bibliography{Bibl}
\bibliographystyle{IEEEtran}
%\begin{thebibliography}{00}
% \bibitem{LDF} LDF
% \bibitem{RAND} Universal randomized - REMIX
% \bibitem{ORIGMODBUF} Li, Fei, Jay Sethuraman, and Clifford Stein. "An optimal online algorithm for packet scheduling with agreeable deadlines." Proceedings of the sixteenth annual ACM-SIAM symposium on Discrete algorithms. Society for Industrial and Applied Mathematics, 2005.
% \end{thebibliography}

\end{document}